\documentclass[journal,a4paper]{IEEEtran}
\usepackage{amssymb}
\usepackage{amsmath}
\usepackage{amsfonts}
\usepackage{float}
\usepackage{graphics}
\usepackage[tight,footnotesize]{subfigure}
\usepackage{rotating}
\usepackage{algorithm}
\usepackage{algorithmic}
\usepackage{enumerate}

\newcommand*{\QEDA}
{\hfill\ensuremath{\blacksquare}}
\DeclareMathSizes{10}{9}{7}{5}
\newtheorem{theo}{Theorem}

\newtheorem{theorem}{Theorem}

\newtheorem{definition}[theorem]{Definition}

\newtheorem{lemma}{Lemma}

\newtheorem{remark}{Remark}
\newtheorem{example}{Example}

\begin{document}

\title{Unshared Secret Key Cryptography}
\pubid{}
\specialpapernotice{}
\author{Shuiyin~Liu,~Yi~Hong,~and~Emanuele~Viterbo\thanks{%
S.~Liu, Y.~Hong and E.~Viterbo are with the Department of Electrical and
Computer Systems Engineering, Monash University, Clayton, VIC 3800,
Australia (e-mail: \{shuiyin.liu, yi.hong, emanuele.viterbo\}@monash.edu).
This work was performed at the Monash Software Defined Telecommunications
Lab and the authors were supported by the Monash Professorial Fellowship,
2013 Monash Faculty of Engineering Seed Funding Scheme, and the Australian
Research Council Discovery Project with ARC DP130100336.}}
\maketitle

\begin{abstract}
Current security techniques can be implemented with either secret key
exchange or physical layer wiretap codes. In this work, we investigate an
alternative solution for MIMO wiretap channels. Inspired by the \emph{%
artificial noise} (AN) technique, we propose the unshared secret key (USK)
cryptosystem, where the AN is redesigned as a one-time pad secret key
aligned within the null space between transmitter and legitimate receiver.
The proposed USK cryptosystem is a new physical layer cryptographic scheme,
obtained by combining traditional network layer cryptography and physical
layer security. Unlike previously studied artificial noise techniques,
rather than ensuring non-zero secrecy capacity, the USK is valid for an
infinite lattice input alphabet and guarantees Shannon's \emph{ideal secrecy}
and \emph{perfect secrecy}, without the need of secret key exchange. We then
show how ideal secrecy can be obtained for finite lattice constellations
with an arbitrarily small outage.
\end{abstract}

\begin{IEEEkeywords}
perfect secrecy, ideal secrecy, secret key, physical layer security,
MIMO wiretap channel.
\end{IEEEkeywords}

\section{Introduction}

The broadcast characteristics of wireless communication systems are
struggling to provide security and privacy. Research on secure communication
falls into two categories: network layer cryptography and physical layer
security. The former assumes that the physical layer provides error-free
data links, in which security depends on encryption. In the latter, the
strategy is to use the characteristics of wireless channels to protect the
secret data from eavesdropping without the need of encryption. Despite the
differences between these categories, both are rooted in Shannon's \emph{%
perfect secrecy} \cite{Shannon46}, which is defined as the mutual
information $I(\mathbf{u}$; $\mathbf{y)}=0$; that is, the secret message $%
\mathbf{u}$ and the eavesdropper's received message $\mathbf{y}$ are
mutually independent. Perfect secrecy requires one-time pad secret key $%
\mathbf{v}$ \cite{Shannon46}. A weaker version of perfect secrecy is \emph{%
ideal secrecy} \cite{Shannon46}, in which no matter how much of $\mathbf{y}$
is intercepted, there is no unique solution of $\mathbf{u}$ and $\mathbf{v}$
but many solutions of comparable probability.

Wyner \cite{Wyner75} introduced physical layer security by replacing
the shared secret key in Shannon's model with channel noise,
achieving \emph{
weak secrecy} ($\lim\limits_{n\rightarrow \infty }\frac{1}{n}I(\mathbf{u}$; $%
\mathbf{y)}=0$) through channel coding as the codeword length $n$ goes to
infinity. Csisz\'{a}r subsequently proposed \emph{strong secrecy} \cite%
{Csiszar96} based upon $\lim\limits_{n\rightarrow \infty }I(\mathbf{u}$; $%
\mathbf{y)}=0$, which further reduced information leakage. These pioneering
results require that the intended receiver has a better channel than the
eavesdropper, leading to a long line of research that relies on noise or
fading to degrade the eavesdropper's channel. Here, the \emph{secrecy
capacity} is defined as a measure of the transmission rate, below which the
eavesdropper can recover no information \cite{Hellman78}. For Gaussian
wiretap channels with a helping interferer, Tang \emph{et al}. \cite{poor08}
studied achievable secrecy rate and secrecy capacity. For wireless fading
channels and multiple-input multiple-output (MIMO) wiretap channels, Gopala
\emph{et al}. \cite{Gopala08} and Liu \emph{et al}. \cite{Shamai09} derived
secrecy capacities. In the context of wiretap code design, polar codes
achieving strong secrecy over discrete memoryless channels have been
proposed in \cite{Vardy11}. For Gaussian wiretap channels, nested lattice
codes achieving strong secrecy were proposed in \cite{cong12IT}. The impact
of finite code length and finite constellations on Eve's equivocation rate
was studied in \cite{CWW11}. Physical layer security schemes, in general,
require an infinite-length wiretap code to approach the secrecy capacity;
this limits the applicability of these schemes to practical communication
systems.

In contrast to physical layer security, traditional cryptographic techniques
can protect the secret message, even when secrecy capacity is zero. Its aim
is to achieve \emph{semantic secrecy} \cite{Goldwasser82}, so that it is
physically infeasible to extract any information about $\mathbf{u}$ due to
the very high computational complexity involved. The most widely used
cryptographic technology is public-key cryptography \cite{Hellman76}, which
requires two separate keys: a public key that encrypts the plaintext and a
secret key that decrypts the ciphertext. An example is the NTRU cryptosystem
\cite{Silverman98}, where the secret key is based on a short vector of a
convolutional modular lattice, $\Lambda $, and the public key corresponds to
the Hermite normal form basis of $\Lambda $ \cite{Micciancio01hnf}. For
wiretap channels, public-key cryptography has been extensively studied in
\cite{Csiszar93,Bennett95}, which focus on issues of secret-key generation
and distribution problems. Bloch \emph{et al}. \cite{Bloch08} showed how to
implement secret-key agreement using low-density parity-check (LDPC) codes.
In \cite{Falahati08}, turbo codes are introduced to speed up the encryption
and decryption processes of the advanced encryption standard (AES)
cryptosystems. Although traditional cryptographic techniques can be applied
independently to communication channels, the exchange of secret keys between
transmitter and intended receiver is required. A significant challenge is to
reduce the risk of key disclosure during its distribution.

Despite the similarities between cryptography and physical layer security,
and the potential for major advances in cryptography through combining their
advantages, the theoretical connections between them have not yet been
investigated. One direction has been to add controlled interference at the
eavesdropper side -- that is, to jam the eavesdropper at the physical layer.
This idea extends previous studies that were limited to the assumptions on
the eavesdropper's channel noise. In the literature, it is commonly assumed
that the transmitted message and jamming signal follow a multivariate
Gaussian distribution. The standard strategy of existing jamming techniques,
such as artificial noise (AN) \cite{Goel08} and the cooperative jammer \cite%
{Fakoorian11}, is to ensure theoretical non-zero secrecy capacity. In \cite%
{Liu13Letter}, we proposed a variant of AN using a finite \textit{M}-QAM
alphabet, called \emph{practical secrecy} (PS) scheme, where, instead of
increasing the secrecy rate with AN, the eavesdropper's error probability is
maximized.

In this work, we analyze the security of the PS scheme from an information
theoretical perspective. This theoretical advance shows that the PS scheme
is \emph{de facto} an \emph{un}shared secret key (USK) cryptosystem, where
AN serves as an unshared one-time pad secret key. The result is a
development of our understanding of the benefits of AN, with a cryptographic
perspective. We show that the USK provides Shannon's ideal secrecy, with no
secret key exchange, under Goel \emph{et al.}'s assumptions on the physical
channels that enable the use of the AN scheme.

Our work differs from previous studies of AN \cite{Goel08,Khisti10SO},
because it puts forward four new aspects that were not previously accounted
for:

\begin{enumerate}
\item \emph{Shannon's secrecy}: we aim at achieving Shannon's ideal secrecy
and perfect secrecy, rather than ensuring non-zero secrecy capacity. We show
that perfect secrecy is achieved in the high-power AN limit.

\item \emph{Finite alphabet based on QAM signaling}: with practical
perspective, we use finite input alphabets rather than the Gaussian input.

\item \emph{Artificial noise}: we have no special requirement of the
distribution of AN; that is, not necessarily Gaussian.

\item \emph{Secrecy outage}: we show that Shannon's ideal secrecy can be
achieved for finite signal constellations with an arbitrarily small outage
probability.
\end{enumerate}

Section II presents the system model. Section III and IV describe
the USK cryptosystem with infinite lattice constellations. Section V
and VI analyze the USK cryptosystem with finite lattice
constellations. Section VII provides a discussion on open questions.
Section VIII sets out the theoretical and applied conclusions. The
Appendix contains the proofs of the theorems.

\textit{Notation:} Matrices and column vectors are denoted by upper and
lowercase boldface letters, and the Hermitian transpose, inverse,
pseudo-inverse of a matrix $\mathbf{B}$ by $\mathbf{B}^{H}$, $\mathbf{B}%
^{-1} $, and $\mathbf{B}^{\dagger }$, respectively. The inner product in the
Euclidean space between vectors $\mathbf{u}$ and $\mathbf{v}$ is defined as $%
\langle \mathbf{u},\mathbf{v}\rangle =\mathbf{u}^{T}\mathbf{v}$, and the
Euclidean length $\Vert \mathbf{u}\Vert =\sqrt{\langle \mathbf{u},\mathbf{u}%
\rangle }$. The Frobenius norm of matrix $\mathbf{A}$ is denoted by $%
\left\Vert \mathbf{A}\right\Vert _{F}$. Let $\left\{ X_{n},X\right\} $ be
defined on the same probability space. We write $X_{n}\overset{a.s.}{%
\rightarrow }X$ if $X_{n}$ converges to $X$ almost surely or with
probability one.

We use the standard asymptotic notation $f\left( x\right) =O\left( g\left(
x\right) \right) $ when $\lim \sup\limits_{x\rightarrow \infty
}|f(x)/g(x)|<\infty $. $\mathbf{0}_{m\times n}$ denotes an $m\times n$ null
matrix. $\mathbf{I}_{n}$ denotes the identity matrix of size $n$. We write $%
\triangleq $ for equality in definition. \textrm{vol}$(S)$ denotes the
Euclidean volume of $S$. The cardinality of a set $A$ is defined as $%
\left\vert A\right\vert $.

A circularly symmetric complex Gaussian random variable $x$ with variance $%
\sigma ^{2}$\ is defined as $x\backsim \mathcal{N}_{\mathbb{C}}(0,\sigma
^{2})$. A Chi-squared distributed random variable $x$ with $k$ degrees of
freedom is defined as $x\backsim \mathcal{X}^{2}(k)$. The gamma function is
represented by $\Gamma (x)$. The real, complex, integer and complex integer
numbers are denoted by $\mathbb{R}$, $\mathbb{C}$, $\mathbb{Z}$, and $%
\mathbb{Z}\left[ i\right] $, respectively. E$(x)$ and Var$(x)$ represent the
mean and variance of the random variable $x$. $\Re (\cdot )$ and $\Im (\cdot
)$ represent real and imaginary parts of a complex number. $H(\cdot )$, $%
H(\cdot |\cdot )$ and $I(\cdot )$ represent entropy, conditional entropy and
mutual information, respectively.

\section{System Model}

The MIMO wiretap system model is given as follows. The number of antennas at
the transmitter (Alice), the intended receiver (Bob), and the passive
eavesdropper (Eve) are denoted by $N_{\text{A}}$, $N_{\text{B}}$, and $N_{%
\text{E}}$, respectively. Alice would like to communicate with Bob with
arbitrarily low probability of error, while maintaining privacy and
confidentiality. Alice transmits the information signal $\mathbf{x}$, and
Bob and Eve receive $\mathbf{z}$ and $\mathbf{y}$, respectively, given by%
\begin{eqnarray}
\mathbf{z} &=&\mathbf{H}\mathbf{x}+\mathbf{n}_{\text{B}}\text{,}
\label{secrecy_signal} \\
\mathbf{y} &=&\mathbf{G}\mathbf{x}+\mathbf{n}_{\text{E}}\text{,}
\label{Eve_signal}
\end{eqnarray}%
where $\mathbf{H}\in \mathbb{C}^{N_{\text{B}}\times N_{\text{A}}}$ and $%
\mathbf{G}\in \mathbb{C}^{N_{\text{E}}\times N_{\text{A}}}$ are the channel
matrices of Bob and Eve. We assume that all the channel matrix elements are
i.i.d. $\mathcal{N}_{\mathbb{C}}(0$, $1)$ random variables (i.e., Bob and
Eve are not co-located). We assume that the noise vectors $\mathbf{n}_{\text{%
B}}$ and $\mathbf{n}_{\text{E}}$ have i.i.d. $\mathcal{N}_{\mathbb{C}}(0$, $%
\sigma _{\text{B}}^{2})$ and $\mathcal{N}_{\mathbb{C}}(0$, $\sigma _{\text{E}%
}^{2})$ components, respectively.

In this work, we assume that

\begin{enumerate}
\item Alice knows the realization of $\mathbf{H}$.

\item Alice only knows the statistics of $\mathbf{G}$, which varies in each
transmission.

\item Eve knows the realizations of $\mathbf{H}$ and $\mathbf{G}$.
\end{enumerate}

No assumption is needed about the statistics of $\mathbf{H}$ during
transmission, since its realization is known to Alice and Eve.

Our secure transmission strategy is based on the artificial noise scheme
\cite{Goel08} and the practical secrecy scheme \cite{Liu13Letter}, which are
summarized below.

\subsection{Artificial Noise Scheme}

In the AN scheme \cite{Goel08}, $N_{\text{B}}$ is assumed to be smaller than
$N_{\text{A}}$, thus $\mathbf{H}$ has a non-trivial null space with an
orthonormal basis given by columns of the matrix $\mathbf{Z}=\mbox{null}(%
\mathbf{H})\in \mathbb{C}^{N_{\text{A}}\times (N_{\text{A}}-N_{\text{B}})}$,
i.e.,%
\begin{equation}
\mathbf{HZ=0}_{N_{\text{B}}\times N_{\text{B}}}\text{.}
\end{equation}%
Let $\mathbf{u}\in \mathbb{C}^{N_{\text{B}}\times 1}$ be the transmitted
vector carrying the information, and let $\mathbf{v}\in \mathbb{C}^{(N_{%
\text{A}}-N_{\text{B}})\times 1}$ represent the \textquotedblleft artificial
noise\textquotedblright\ generated by Alice but is unknown to Bob and Eve.

Alice performs \emph{SVD precoding} and transmits%
\begin{equation}
\mathbf{x=V}\left[
\begin{array}{c}
\mathbf{{{\mathbf{u}}}} \\
\mathbf{v}%
\end{array}%
\right] =\mathbf{\mathbf{V}}_{1}\mathbf{{{\mathbf{u}}}+Zv}\text{,}
\label{T1}
\end{equation}%
where the columns of $\mathbf{V=[V}_{1}$, $\mathbf{Z]}$ are the
right-singular vectors of $\mathbf{H}$ (i.e., $\mathbf{H}=\mathbf{U}\mathbf{%
\Lambda }\mathbf{V}^{H}$, where $\mathbf{U}\in \mathbb{C}^{N_{\text{B}%
}\times N_{\text{B}}}$, $\mathbf{\Lambda }\in \mathbb{C}^{N_{\text{B}}\times
N_{\text{A}}}$, $\mathbf{V}\in \mathbb{C}^{N_{\text{A}}\times N_{\text{A}}}$%
, $\mathbf{U}^{H}\mathbf{U}=\mathbf{I}_{N_{\text{B}}}$, $\mathbf{V}^{H}%
\mathbf{V}=\mathbf{I}_{N_{\text{A}}}$).

Due to the orthogonality between $\mathbf{\mathbf{\mathbf{V}}}_{1}$ and $%
\mathbf{Z}$, the total transmission power $||\mathbf{x}||^{2}$ can be
written as%
\begin{equation}
||\mathbf{x}||^{2}=||\mathbf{u}||^{2}+||\mathbf{v}||^{2}\text{.}
\label{Total_power}
\end{equation}

Alice has an average transmit power constraint $P$,%
\begin{equation}
P\geq \text{E}(||\mathbf{x}||^{2})=\text{E}(||\mathbf{u}||^{2})+\text{E}(||%
\mathbf{v}||^{2})\text{.}
\end{equation}

The AN scheme in \cite{Goel08} is based on the assumptions below:

\begin{enumerate}
\item $\mathbf{u}$ and $\mathbf{v}$ are assumed to be Gaussian random
vectors.

\item $N_{\text{A}}>N_{\text{B}}$, $N_{\text{A}}>N_{\text{E}}$ and $N_{\text{%
E}}\geq N_{\text{B}}$
\end{enumerate}

The condition $N_{\text{E}}\geq N_{\text{B}}$ guarantees that Eve has at
least the same number of degree of freedom as Bob. This puts Eve in the
position of not losing \emph{a-priori} any information that Bob could
receive.

Equations (\ref{secrecy_signal}) and (\ref{Eve_signal}) can then be
rewritten as%
\begin{align}
\mathbf{z}& =\mathbf{H\mathbf{\mathbf{V}}}_{1}\mathbf{u}+\mathbf{n}_{\text{B}%
}\text{,}  \label{sec_mod2} \\
\mathbf{y}& =\mathbf{G\mathbf{\mathbf{V}}}_{1}\mathbf{u}+\mathbf{G}\mathbf{Z}%
\mathbf{v}+\mathbf{n}_{\text{E}}\text{.}  \label{Eve_mod2}
\end{align}%
and show that $\mathbf{v}$ only degrades Eve's reception, but not Bob's.

The purpose of the AN scheme is to degrade Eve's channel, so that the
secrecy capacity is positive \cite{Goel08}. Like other wiretap schemes, to
achieve the secrecy capacity, explicit wiretap codes are required. A \emph{%
strong secrecy rate} $R$ is achievable if there exist a sequence of wiretap
codes $\left\{ \mathcal{C}_{n}\right\} $ of increasing length $n$ and rate $%
R $, such that both Bob's error probability and the amount of information
obtained by Eve approach zero when $n\rightarrow \infty $ \cite%
{Csiszar96,cong12IT}, i.e.,%
\begin{align*}
& \lim_{n\rightarrow \infty }\Pr \left\{ \mathbf{\hat{u}\neq u}\right\} =0%
\text{,}  \tag{reliability} \\
& \lim_{n\rightarrow \infty }I(\mathbf{u};\mathbf{y)}=0\text{,}
\tag{strong
secrecy}
\end{align*}%
where $\mathbf{\hat{u}}$ represents Bob's estimation of $\mathbf{u}$.

\subsection{Practical Secrecy Scheme}

Rather than attempting to increase secrecy rate, in \cite{Liu13Letter}, we
proposed a variant of the AN scheme, named \emph{practical secrecy} (PS)
scheme, where Eve's error probability is maximized. Although the
transmission model is the same as that of AN, the most important difference
lies in the distributions of $\mathbf{u}$ and $\mathbf{v}$:

\begin{enumerate}
\item $M$-QAM transmitted symbols: $\mathbf{u}\in \mathcal{Q}^{N_{\text{B}}}$
with uniform distribution, where $\Re (\mathcal{Q}\mathbf{)}=\Im (\mathcal{Q}%
\mathbf{)}=\{-\sqrt{M}+1$, $-\sqrt{M}+3$, ..., $\sqrt{M}-1\}$.

\item There is no requirement on the distribution of $\mathbf{v}$.
\end{enumerate}

Different from the AN scheme, where the achievability of security is based
on an infinite-length wiretap code, the PS scheme \cite{Liu13Letter} is
designed for practical communication systems, that make use of finite input
alphabets based on $M$-QAM transmitted symbols. The aim is to ensure that
Eve's block error probability approaches one with minimum distance decoding,
(e.g. sphere decoder), rather than strong secrecy. However, this security
criterion is not satisfactory from an information-theoretic security
viewpoint, as it may not ensure security for all information bits within a
message.

\subsection{Proposed AN Scheme}

Different from the original AN scheme \cite{Goel08}, in this work, we set a
peak AN power constraint,%
\begin{equation}
P_{\text{v}}\geq ||\mathbf{v}||^{2}\text{.}
\end{equation}%
This peak power constraint is essential to prove the secrecy of USK, as
detailed in Section III-A.

Moreover, we consider two lattice constellation models:

\begin{enumerate}
\item Infinite constellations with average power constraint

\item Finite constellations with peak power constraint
\end{enumerate}

We focus on information theoretic security, hence, our analysis will focus
on Eve's equivocation $H(\mathbf{{{\mathbf{u}}}|y})$.

Throughout the paper, we consider the \emph{worst-case} scenario (for Alice)
that Eve's channel is noiseless, i.e.,%
\begin{equation}
\mathbf{y}=\mathbf{G\mathbf{\mathbf{V}}}_{1}\mathbf{u}+\mathbf{G}\mathbf{Z}%
\mathbf{v}\text{.}  \label{WC_S}
\end{equation}%
Using \emph{Data Processing Inequality}, it is simple to show that Eve's
channel noise can only increase her equivocation:%
\begin{equation}
H(\mathbf{{{\mathbf{u}}}|G\mathbf{\mathbf{V}}}_{1}\mathbf{u}+\mathbf{G}%
\mathbf{Z}\mathbf{v})\leq H(\mathbf{{{\mathbf{u}}}|G\mathbf{\mathbf{V}}}_{1}%
\mathbf{u}+\mathbf{G}\mathbf{Z}\mathbf{v+n}_{\text{E}})\text{.}
\end{equation}

We further consider the \emph{worst-case} scenario (for Alice) that Eve's
antenna array elements are uncorrelated, i.e., the columns of $\mathbf{G}$
are zero-mean independent complex Gaussian vectors with an identity
covariance matrix.

For a general complex Gaussian random matrix $\mathbf{\hat{G}}$ with an
arbitrary non-singular covariance matrix $\Sigma $ (which is the covariance
matrix\ of Eve's antenna array), we can write%
\begin{equation}
\mathbf{\hat{G}}=\Sigma ^{1/2}\mathbf{G}\text{.}
\end{equation}%
Using \emph{Data Processing Inequality}, it is simple to show that Eve's
antenna correlation can only increase her equivocation:%
\begin{equation}
H(\mathbf{{{\mathbf{u}}}|G\mathbf{\mathbf{V}}}_{1}\mathbf{u}+\mathbf{G}%
\mathbf{Z}\mathbf{v})\leq H(\mathbf{{{\mathbf{u}}}|}\Sigma ^{1/2}\mathbf{G%
\mathbf{\mathbf{V}}}_{1}\mathbf{u}+\Sigma ^{1/2}\mathbf{G}\mathbf{Z}\mathbf{v%
})\text{.}
\end{equation}

\begin{remark}
Throughout this paper, the proposed security analysis of USK scheme is valid
for a complex Gaussian random matrix $\mathbf{G}$ with an arbitrary
non-singular covariance matrix $\Sigma $. The extension to USK of other
distributed random matrix $\mathbf{G}$ will be studied in our future work.
\end{remark}

\subsection{Shannon's Secrecy}

We consider a cryptosystem where a sequence of $K$ messages $\left\{ \mathbf{%
m}_{i}\right\} _{1}^{K}$ are enciphered into the cryptograms $\left\{
\mathbf{y}_{i}\right\} _{1}^{K}$ using a sequence of secret keys $\left\{
k_{i}\right\} _{1}^{K}$. We recall from \cite{Shannon46} the definition of
Shannon's ideal secrecy and perfect secrecy.

\begin{definition}
\label{Def_IS copy(1)}A secrecy system is \emph{ideal} when%
\begin{eqnarray}
\lim\limits_{K\rightarrow \infty }H(\left\{ \mathbf{m}_{i}\right\} _{1}^{K}%
\mathbf{|}\left\{ \mathbf{y}_{i}\right\} _{1}^{K}) &\neq &0\text{,}  \notag
\\
\lim\limits_{K\rightarrow \infty }H(\left\{ k_{i}\right\} _{1}^{K}\mathbf{|}%
\left\{ \mathbf{y}_{i}\right\} _{1}^{K}) &\neq &0\text{.}
\end{eqnarray}
\end{definition}

Shannon explained the concept of ideal secrecy in \cite{Shannon46} as:
\textquotedblleft No matter how much material is intercepted, there is not a
unique solution but many of comparable probability.\textquotedblright\ It
was discussed in \cite{Geffe65} how a system achieving ideal secrecy is
indeed unbreakable.

\begin{definition}
\label{Def_PS copy(1)}A secrecy system is \emph{perfect} when%
\begin{equation}
H(\left\{ \mathbf{m}_{i}\right\} _{1}^{K}\mathbf{|}\left\{ \mathbf{y}%
_{i}\right\} _{1}^{K})=H(\left\{ \mathbf{m}_{i}\right\} _{1}^{K})\text{.}
\end{equation}
\end{definition}

In the special case that $\left\{ \mathbf{m}_{i}\right\} _{1}^{K}$ and $%
\left\{ k_{i}\right\} _{1}^{K}$ are mutually independent, using the entropy
chain rule, we have%
\begin{equation}
H(\left\{ \mathbf{m}_{i}\right\} _{1}^{K})=\sum_{i=1}^{K}H(\mathbf{m}_{i})
\label{HM}
\end{equation}%
\begin{equation}
H(\left\{ \mathbf{m}_{i}\right\} _{1}^{K}\mathbf{|}\left\{ \mathbf{y}%
_{i}\right\} _{1}^{K})=\sum_{i=1}^{K}H(\mathbf{m}_{i}\mathbf{|y}_{i})\text{,}
\label{Ideal_m1}
\end{equation}%
\begin{equation}
H(\left\{ k_{i}\right\} _{1}^{K}\mathbf{|}\left\{ \mathbf{y}_{i}\right\}
_{1}^{K})=\sum_{i=1}^{K}H(k_{i}\mathbf{|y}_{i})\text{.}  \label{Ideal_m2}
\end{equation}

From (\ref{Ideal_m1}) and (\ref{Ideal_m2}), ideal secrecy is achieved if $H(%
\mathbf{m}_{i}\mathbf{|y}_{i})\neq 0$ and $H(k_{i}\mathbf{|y}_{i})\neq 0$
for one of any $i$. To protect all the messages, in this work, we use a
slightly stronger condition as our design criterion for ideal secrecy, given
by

\begin{definition}
\label{Def_S_IS}If $\left\{ \mathbf{m}_{i}\right\} _{1}^{K}$ and $\left\{
k_{i}\right\} _{1}^{K}$ are mutually independent, a secrecy system is \emph{%
ideal} when%
\begin{equation}
H(\mathbf{m}_{i}\mathbf{|y}_{i})\neq 0\text{ and }H(k_{i}\mathbf{|y}%
_{i})\neq 0\text{, for all }i\text{.}  \label{SE1}
\end{equation}
\end{definition}

From (\ref{HM}) and (\ref{Ideal_m1}), perfect secrecy is achieved when%
\begin{equation}
H(\mathbf{m}_{i}\mathbf{|y}_{i})=H(\mathbf{m}_{i})\text{, for all }i\text{.}
\label{SE2}
\end{equation}

An overview of measures on information-theoretic security can be found in
\cite{Liu_inv}.

\subsection{Lattice Preliminaries}

To describe our scheme, it is convenient to introduce some lattice
preliminaries. An $n$-dimensional \emph{complex lattice} $\Lambda _{\mathbb{C%
}}$ in a complex space $\mathbb{C}^{m}$ ($n\leq m$) is the discrete set
defined by:%
\begin{equation*}
\Lambda _{\mathbb{C}}=\left\{ \mathbf{Bu}\text{: }\mathbf{u\in }\text{ }%
\mathbb{Z}\left[ i\right] ^{n}\right\} \text{,}
\end{equation*}%
where the \emph{basis} matrix $\mathbf{B=}\left[ \mathbf{b}_{1}\cdots
\mathbf{b}_{n}\right] $ has linearly independent columns.

$\Lambda _{\mathbb{C}}$ can also be easily represented as $2n$-dimensional
real lattice $\Lambda _{\mathbb{R}}$ \cite{BK:Conway93}. In what follows, we
introduce some lattice parameters of $\Lambda _{\mathbb{C}}$, which have a
corresponding value for $\Lambda _{\mathbb{R}}$. The \emph{Voronoi region}
of $\Lambda _{\mathbb{C}}$, defined by%
\begin{equation*}
\mathcal{V}_{i}\left( \Lambda _{\mathbb{C}}\right) =\left\{ \mathbf{y}\in
\mathbb{C}^{m}\text{: }\Vert \mathbf{y}-\mathbf{x}_{i}\Vert \leq \Vert
\mathbf{y}-\mathbf{x}_{j}\Vert ,\forall \text{ }\mathbf{x}_{i}\neq \mathbf{x}%
_{j}\right\} ,
\end{equation*}%
gives the nearest neighbor decoding region of lattice point $\mathbf{x}_{i}$.

The volume of any $\mathcal{V}_{i}\left( \Lambda _{\mathbb{C}}\right) $,
defined as vol$(\Lambda _{\mathbb{C}})\triangleq |\det (\mathbf{B}^{H}%
\mathbf{B})|$, is equivalent to the volume of the corresponding real lattice.

The \emph{effective radius} of $\Lambda _{\mathbb{C}}$, denoted by $r_{\text{%
eff}}(\Lambda _{\mathbb{C}})$, is the radius of a sphere$\,$of volume vol$%
(\Lambda _{\mathbb{C}})$ \cite{Zamir08}. For large $n$, it is approximately%
\begin{equation}
r_{\text{eff}}(\Lambda _{\mathbb{C}})\approx \sqrt{n/(\pi e)}\text{vol}%
(\Lambda _{\mathbb{C}})^{\frac{1}{2n}}.  \label{r_eff}
\end{equation}

\section{Unshared Secret Key Cryptosystem With Infinite Constellations}

In this section, we consider the system model with an infinite lattice
constellations, satisfying the average transmit power constraint. This
provides the theoretical basis for unshared secret key cryptosystems.

\subsection{Encryption}

We consider a sequence of $K$ mutually independent messages $\left\{ \mathbf{%
m}_{i}\right\} _{1}^{K}$, where each one is mapped to a transmitted vector $%
\mathbf{{{\mathbf{u}}}}\in {\mathbb{Z}\left[ i\right] ^{N_{\text{B}}}}$. The
probability distribution of $\mathbf{{{\mathbf{u}}}}$ can be arbitrary, but
has finite E$(||\mathbf{{{\mathbf{u}}}}||^{2})$. To secure the $K$
transmitted vectors $\left\{ \mathbf{u}_{i}\right\} _{1}^{K}$, Alice
enciphers $\left\{ \mathbf{u}_{i}\right\} _{1}^{K}$ into the cryptograms $%
\left\{ \mathbf{y}_{i}\right\} _{1}^{K}$ using a sequence of mutually
independent secret keys $\left\{ \mathbf{v}_{i}\right\} _{1}^{K}$. We assume
that $\left\{ \mathbf{v}_{i}\right\} _{1}^{K}$ and $\left\{ \mathbf{u}%
_{i}\right\} _{1}^{K}$ are mutually independent, and $\left\{ \mathbf{G}%
_{i}\right\} _{1}^{K}$ are mutually independent Gaussian random matrices. No
assumption is needed about the statistics of $\left\{ \mathbf{H}_{i}\right\}
_{1}^{K}$ across the $K$ channel uses, since its realization is known to
both Alice and Eve.

Since $\left\{ \mathbf{v}_{i}\right\} _{1}^{K}$ and $\left\{ \mathbf{u}%
_{i}\right\} _{1}^{K}$ are mutually independent, from (\ref{SE1}) and (\ref%
{SE2}), we only need to demonstrate the encryption process for one
transmitted vector $\mathbf{u}_{i}$. For simplicity, we drop the subscript $%
i $.

For each $\mathbf{{{\mathbf{u}}}}$, Alice randomly and independently
(without any predefined distribution) chooses a one-time pad secret key $%
\mathbf{v}$, from a ball of radius $\sqrt{P_{\text{v}}}$:%
\begin{equation}
S\triangleq \left\{ \mathbf{v}\in \mathbb{C}^{N_{\text{A}}-N_{\text{B}}}%
\text{: }||\mathbf{v||}^{2}\leq P_{\text{v}}\right\} \text{,}  \label{set}
\end{equation}%
and transmits%
\begin{equation}
\mathbf{x=\mathbf{\mathbf{\mathbf{V}}}}_{1}\mathbf{{{\mathbf{u}}}+Zv}\text{.}
\label{Bob_U}
\end{equation}

In the worst-case scenario, when $\mathbf{n}_{\text{E}}=\mathbf{0}$, Eve
will receive (\ref{WC_S}), i.e.,%
\begin{equation}
\mathbf{y}=\mathbf{G\mathbf{\mathbf{\mathbf{\mathbf{V}}}}}_{1}\mathbf{%
\mathbf{{{\mathbf{u}}}}{{\mathbf{+}}}\tilde{n}}_{\text{v}}\text{,}
\label{Eve_U}
\end{equation}%
where $\mathbf{\tilde{n}}_{\text{v}}=\mathbf{G}\mathbf{Z}\mathbf{v}$.

The signal model (\ref{Eve_U}) can be interpreted as an encryption
algorithm, that is, the secret message $\mathbf{u}$ is encrypted to $\mathbf{%
y}$ using a secret key $\mathbf{v}$, which is not released neither to Bob
nor to Eve.

The message $\mathbf{u}$ is received by Eve as a lattice point (see Fig.~\ref%
{fig_k_th_point}) in:%
\begin{equation}
\Lambda _{\mathbb{C}}=\{\mathbf{G\mathbf{\mathbf{\mathbf{\mathbf{V}}}}}_{1}%
\mathbf{\mathbf{{{\mathbf{u}}}}},\mathbf{u}\in {\mathbb{Z}\left[ i\right]
^{N_{\text{B}}}\}}\text{.}  \label{INF_LT}
\end{equation}%
This enables us to partition the set $S$ into $D$ disjoint subsets $S_{1}$,
..., $S_{D}$, such that%
\begin{equation}
S=\bigcup\limits_{k=1}^{D}S_{k}\text{,}  \label{subset}
\end{equation}%
where%
\begin{equation}
S_{k}\triangleq \left\{ \mathbf{v}\text{: }\mathbf{G\mathbf{\mathbf{\mathbf{%
\mathbf{V}}}}}_{1}\mathbf{\mathbf{{{\mathbf{u}}}}}\in \Lambda _{\mathbb{C}}%
\text{ is the }k^{\text{th}}\text{ closest lattice point to }\mathbf{y}%
\right\} \text{.}  \label{sub_k}
\end{equation}

As shown in Fig.~1, the value of $D$ is determined by%
\begin{equation}
D=\left\vert S_{R_{\max }}\cap \Lambda _{\mathbb{C}}\right\vert \text{,}
\end{equation}%
where $S_{R_{\max }}$ is a sphere centered at $\mathbf{y}$ with radius%
\begin{equation}
R_{\max }(P_{\text{v}})\triangleq \max_{||\mathbf{v||}^{2}\leq P_{\text{v}%
}}\left\Vert \mathbf{G}\mathbf{Z}\mathbf{v}\right\Vert =\sqrt{\lambda _{\max
}P_{\text{v}}}\text{,}  \label{R_max}
\end{equation}%
where $\lambda _{\max }$ is the largest eigenvalue of $(\mathbf{G}\mathbf{Z}%
)^{H}(\mathbf{G}\mathbf{Z})$.

Assuming $\mathbf{v}\in S_{k}$, $1\leq k\leq D$, the signal model (\ref%
{Eve_U}) can be further viewed as an encryption algorithm that encrypts $%
\mathbf{{{\mathbf{u}}}}$ to $\mathbf{y}$ using a one time pad secret key $%
\mathbf{v}$, such that $\mathbf{G\mathbf{\mathbf{\mathbf{\mathbf{V}}}}}_{1}%
\mathbf{\mathbf{{{\mathbf{u}}}}}$ is the $k^{\text{th}}$ closest lattice
point to $\mathbf{y}$.

The security problem lies in how much Eve knows about $k$. The value of $k$
is uniquely determined by the vector $\mathbf{\tilde{n}}_{\text{v}}$. Since
we assume that the realizations of $\mathbf{G}$ and $\mathbf{Z}$ are known
to Eve, $k$ is a function of $\mathbf{v}$. Since $\mathbf{v}$ is randomly
and independently selected by Alice and is never disclosed to anyone, Eve
can neither know its realization nor its distribution. Thus, given $\mathbf{y%
}$, Eve is not able to estimate the distribution of the index $k$.

\begin{remark}
The index $k$ can be interpreted as the \emph{effective} one-time pad secret
key, whose randomness comes from the artificial noise. The effective key
space size is $D$.
\end{remark}

From Eve's perspective, we assume that she knows $P_{\text{v}}$, $R_{\max
}(P_{\text{v}})$, $D$ and the encryption process (\ref{Eve_U}). Based on the
above analysis, given $\mathbf{y}$, Eve only knows that $\mathbf{G\mathbf{%
\mathbf{\mathbf{\mathbf{V}}}}}_{1}\mathbf{\mathbf{{{\mathbf{u}}}}}\in
S_{R_{\max }}\cap \Lambda _{\mathbb{C}}$. Therefore, the posterior
probability that Eve obtains $\mathbf{{{\mathbf{u}}}}$, or equivalently,
finds $k$, from the cryptogram $\mathbf{y}$, is equal to%
\begin{equation}
\Pr \left\{ \mathbf{{{\mathbf{u}}}|y}\right\} =\Pr \left\{ k\mathbf{|y}%
\right\} =\Pr \left\{ \mathbf{{{\mathbf{u}}}|{{\mathbf{u}}}}\in \mathcal{U}%
\right\} \text{,}  \label{Post_Pr}
\end{equation}%
where%
\begin{equation}
\mathcal{U}\triangleq \left\{ \mathbf{u}^{\prime }\text{: }\mathbf{G\mathbf{%
\mathbf{\mathbf{\mathbf{V}}}}}_{1}\mathbf{\mathbf{{{\mathbf{u}}}}}^{\prime
}\in S_{R_{\max }}\cap \Lambda _{\mathbb{C}}\right\} \text{,}
\end{equation}%
and $\left\vert \mathcal{U}\right\vert =D$.

For any $\mathbf{u}^{\prime }\in \mathcal{U}$, using Bayes' theorem, we have%
\begin{eqnarray}
\Pr \left\{ \mathbf{{{\mathbf{u}}}}=\mathbf{u}^{\prime }\mathbf{|{{\mathbf{u}%
}}}\in \mathcal{U}\right\} &=&\frac{\Pr \left\{ \mathbf{{{\mathbf{u}}}}=%
\mathbf{u}^{\prime }\right\} \Pr \left\{ \mathbf{{{\mathbf{u}}}}\in \mathcal{%
U}|\mathbf{{{\mathbf{u}}}}=\mathbf{u}^{\prime }\right\} }{\Pr \left\{
\mathbf{{{\mathbf{u}}}}\in \mathcal{U}\right\} }  \notag \\
&=&\frac{\Pr \left\{ \mathbf{{{\mathbf{u}}}}=\mathbf{u}^{\prime }\right\} }{%
\Pr \left\{ \mathbf{{{\mathbf{u}}}}\in \mathcal{U}\right\} }\text{.}
\label{Post_Pr_d}
\end{eqnarray}%
\

From (\ref{Post_Pr}) and (\ref{Post_Pr_d}), Eve's equivocation is given by%
\begin{equation}
H(\mathbf{{{\mathbf{u}}}|y})=H(k\mathbf{|y})=\sum\limits_{\mathbf{u}^{\prime
}\in \mathcal{U}}\frac{\Pr \left\{ \mathbf{{{\mathbf{u}}}}=\mathbf{u}%
^{\prime }\right\} }{\Pr \left\{ \mathbf{{{\mathbf{u}}}}\in \mathcal{U}%
\right\} }\log \frac{\Pr \left\{ \mathbf{{{\mathbf{u}}}}\in \mathcal{U}%
\right\} }{\Pr \left\{ \mathbf{{{\mathbf{u}}}}=\mathbf{u}^{\prime }\right\} }%
\text{.}  \label{uncer}
\end{equation}

Since%
\begin{equation}
\Pr \left\{ \mathbf{{{\mathbf{u}}}}\in \mathcal{U}\right\} =\sum\limits_{%
\mathbf{u}^{\prime }\in \mathcal{U}}\Pr \left\{ \mathbf{{{\mathbf{u}}}}=%
\mathbf{u}^{\prime }\right\} \text{,}
\end{equation}%
the security level is determined by the cardinality of the set $\mathcal{U}$%
, or more specifically, by the value of $D$:

\begin{enumerate}
\item if $D=1$, then $\Pr \left\{ \mathbf{{{\mathbf{u}}}}\in \mathcal{U}%
\right\} =\Pr \left\{ \mathbf{{{\mathbf{u}}}}=\mathbf{u}^{\prime }\right\} $%
, so that%
\begin{equation}
H(k\mathbf{|y})=H(\mathbf{{{\mathbf{u}}}|y})=0\text{.}  \tag{no security}
\end{equation}

\item if $D\geq 2$, then $\Pr \left\{ \mathbf{{{\mathbf{u}}}}\in \mathcal{U}%
\right\} >\Pr \left\{ \mathbf{{{\mathbf{u}}}}=\mathbf{u}^{\prime }\right\} $%
, so that%
\begin{equation}
H(k\mathbf{|y})=H(\mathbf{{{\mathbf{u}}}|y})>0\text{.}  \tag{ideal secrecy}
\end{equation}

\item as $D\rightarrow \infty $, then $\Pr \left\{ \mathbf{{{\mathbf{u}}}}%
\in \mathcal{U}\right\} \rightarrow 1$, so that%
\begin{equation}
H(k\mathbf{|y})=H(\mathbf{{{\mathbf{u}}}|y})=H(\mathbf{{{\mathbf{u}}}})\text{%
.}  \tag{perfect secrecy}
\end{equation}
\end{enumerate}

\begin{remark}
Different from Shannon's one-time pad cryptosystem, the effective\emph{\ }%
one-time pad secret key $k$ is not shared between Alice and Bob. In
particular, it is independently generated by Alice, but not needed by Bob to
decipher, while it is fully affecting Eve's ability to decipher the original
message. This motivates the name of this cryptosystem as \emph{Unshared
Secret Key} (USK) cryptosystem.
\end{remark}

\begin{figure}[tbp]
\centering\includegraphics[scale=0.65]{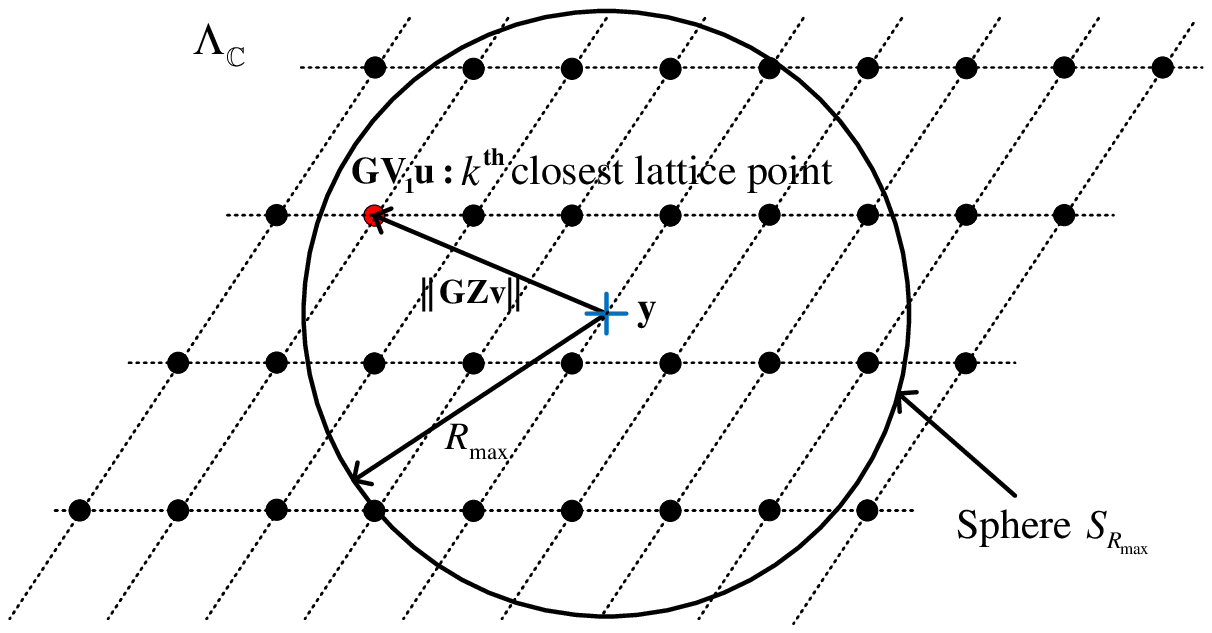}
\caption{ The USK cryptosystem with infinite constellations.}
\label{fig_k_th_point}
\end{figure}

\subsection{Analyzing Eve's Equivocation}

As shown in (\ref{uncer}), Eve's equivocation lies in the value of $D$,
which is known to Eve but not to Alice. We then estimate the value of $D$
from Alice's perspective. According to (\ref{subset}) and (\ref{sub_k}), $D$
is a function of $P_{\text{v}}$, $\mathbf{H}$, and $\mathbf{G}$. Alice knows
only$\ P_{\text{v}}$ and $\mathbf{H}$, while regarding $\mathbf{G}$, she
knows the statistics, but doesn't know the realization. Although Alice
cannot know the exact value of $D$, she is able to estimate its cumulative
distribution function (cdf), denoted by%
\begin{equation}
F_{D}(d\text{, }P_{\text{v}})\triangleq \Pr \left\{ D<d\right\} \text{,}
\label{F_D}
\end{equation}%
where $d$ is a positive integer.

In the next section, we will show that Alice can ensure $F_{D}(d$, $P_{\text{%
v}})\rightarrow 0$ by increasing $P_{\text{v}}$, i.e., she can guarantee
that $D\geq d$, for any given $d$.

\section{The Security of USK with Infinite Constellations}

In this section, we show that the USK with infinite constellations provides
Shannon's ideal secrecy and perfect secrecy. To prove the main theorems, we
first introduce some lemmas.

We first define%
\begin{equation}
\kappa (d)\triangleq d^{1/(2N_{\text{E}})}/\sqrt{\pi }\text{,}  \label{kappa}
\end{equation}%
\begin{equation}
\Delta \left( d\right) \triangleq \frac{\kappa (d)^{2N_{\text{E}}}\text{vol}%
(\Lambda _{\mathbb{C}})}{P_{\text{v}}^{N_{\text{E}}}}\text{,}  \label{delta}
\end{equation}%
where $d$ is an integer and%
\begin{equation}
\text{vol}(\Lambda _{\mathbb{C}})=|\det (\left( \mathbf{GV}_{1}\right)
^{H}\left( \mathbf{GV}_{1}\right) )|\text{.}
\end{equation}

Here, $\mathbf{G}$ is a complex Gaussian random matrix, while $\mathbf{V}%
_{1} $ is deterministic. Thus, $\Delta \left( d\right) $ is a random
variable from Alice perspective. The following two lemmas are used to
evaluate $F_{D}(d$, $P_{\text{v}})$ in (\ref{F_D}).

\begin{lemma}
\label{lem2}If $P_{\text{v}}\geq \rho ^{2}/\Phi ^{2N_{\text{B}}/N_{\text{E}%
}} $ and $\rho >\kappa (d)$, then $\Delta (d)\overset{a.s.}{\rightarrow }0$
as $N_{\text{B}}\rightarrow \infty $, or equivalently,%
\begin{equation}
\Pr \left\{ \Delta (d)>\left( \frac{\rho }{\kappa (d)}\right) ^{-N_{\text{B}%
}}\right\} <O\left( \left( \frac{\rho }{\kappa (d)}\right) ^{-N_{\text{B}%
}}\right)  \label{Pr_inf}
\end{equation}%
where%
\begin{equation}
\Phi =\left[ \frac{(N_{\text{E}}-N_{\text{B}})!}{N_{\text{E}}!}\right] ^{%
\frac{1}{2N_{\text{B}}}}\text{.}  \label{Pi_SVD}
\end{equation}
\end{lemma}

\begin{proof}
See Appendix A.
\end{proof}

We next provide a more accurate expression of the tail distribution of $%
\Delta \left( d\right) $ for finite $N_{\text{B}}$.

\begin{lemma}
\label{lem3}If $P_{\text{v}}\geq \rho ^{2}/\Phi ^{2N_{\text{B}}/N_{\text{E}%
}} $ and $\rho >\kappa (d)$, then%
\begin{equation}
\Pr \left\{ \Delta (d)>\left( \frac{\rho }{\kappa (d)}\right) ^{-N_{\text{B}%
}}\right\} <\Upsilon \left( \frac{\rho }{\kappa (d)}\right) \text{,}
\label{Pr_Fin}
\end{equation}%
where $\kappa (d)$ is given in (\ref{kappa}), $\Phi $ is given (\ref{Pi_SVD}%
), and%
\begin{equation}
\Upsilon (x)=\sum_{i=1}^{N_{\text{B}}}\left( xe^{1-x}\right) ^{N_{\text{E}%
}-i+1}\text{.}  \label{BR_SVD_1}
\end{equation}
\end{lemma}

\begin{proof}
See Appendix B.
\end{proof}

\begin{remark}
From (\ref{BR_SVD_1}), it is easy to see that $\Upsilon (x)$ is
monotonically decreasing function. Let
\begin{equation}
N\triangleq N_{\text{E}}-N_{\text{B}}+1\text{,}  \label{N}
\end{equation}%
then, as $x\rightarrow \infty $, we have%
\begin{equation}
\Upsilon (x)=O\left( \left( x^{-1}e^{x}\right) ^{-N}\right) =O\left(
e^{-xN}\right) \text{.}  \label{NLP0}
\end{equation}
\end{remark}

Lemmas \ref{lem2} and \ref{lem3} enable us to prove the following lemma.

\begin{lemma}
\label{lem1}If $P_{\text{v}}\geq \rho ^{2}/\Phi ^{2N_{\text{B}}/N_{\text{E}%
}} $ and $\rho >\kappa (d)$, $F_{D}(d,P_{\text{v}})\rightarrow 0$ as $N_{%
\text{B}}\rightarrow \infty $, or equivalently,%
\begin{equation}
F_{D}(d,P_{\text{v}})<O\left( \left( \frac{\rho }{\kappa (d)}\right) ^{-N_{%
\text{B}}}\right) \text{,}  \label{L3_1}
\end{equation}%
and for finite $N_{\text{B}}$, we have%
\begin{equation}
F_{D}(d,P_{\text{v}})<\left( \frac{\rho }{\kappa (d)}\right) ^{-N_{\text{B}%
}}+\Upsilon \left( \frac{\rho }{\kappa (d)}\right) \text{,}
\end{equation}%
where $\kappa (d)$ is given in (\ref{kappa}), $\Phi $ is given in (\ref%
{Pi_SVD}), and $\Upsilon (x)$ is given in (\ref{BR_SVD_1}).
\end{lemma}

\begin{proof}
See Appendix C.
\end{proof}

\subsection{Achieving Ideal Secrecy}

From (\ref{SE1}) and\ the discussion following (\ref{uncer}), ideal secrecy
is achieved when $D\geq 2$. Lemma \ref{lem1} enable us to prove the
following equivalent theorem about achieving ideal secrecy.

\begin{theo}
\label{Th1}If $P_{\text{v}}>\kappa (d)^{2}/\Phi ^{2N_{\text{B}}/N_{\text{E}%
}} $ and $d\geq 2$, as $N_{\text{B}}\rightarrow \infty $,%
\begin{equation}
D\overset{a.s.}{\geq }d\text{,}
\end{equation}%
where $\kappa (d)$ is given in (\ref{kappa}) and $\Phi $ is given in (\ref%
{Pi_SVD}).
\end{theo}

\begin{proof}
From (\ref{F_D}) and (\ref{L3_1}), it is straightforward to see that $\Pr
\left( D<d\right) \rightarrow 0$ as $N_{\text{B}}\rightarrow \infty $.
\end{proof}

Theorem \ref{Th1} shows that for the USK, Eve cannot find a unique solution $%
\mathbf{u}$, since $D$ is almost surely greater than $2$.

We next estimate the \emph{secrecy outage probability} when $N_{\text{B}}$
is finite, defined by%
\begin{equation}
P_{\text{out}}(d)\triangleq \Pr \left\{ D<d\right\} \text{,}  \label{P_out}
\end{equation}%
for any $d\geq 2$.

\begin{theo}
\label{Th2}Let $N_{\min }=\min \left\{ N\text{, }N_{\text{B}}\right\} $,
where $N$ is given in (\ref{N}). If%
\begin{equation}
P_{\text{v}}\geq \varepsilon ^{-2/N_{\min }}\kappa (d)^{2}/\Phi ^{2N_{\text{B%
}}/N_{\text{E}}}
\end{equation}%
and $d\geq 2$, then%
\begin{equation}
P_{\text{out}}(d)<O(\varepsilon )\text{,}
\end{equation}%
for any arbitrarily small $\varepsilon >0$, i.e., ideal secrecy is achieved
with probability $1-O(\varepsilon )$, where $\kappa (d)$ is given in (\ref%
{kappa}) and $\Phi $ is given in (\ref{Pi_SVD}).
\end{theo}

\begin{proof}
See Appendix D.
\end{proof}

Theorem \ref{Th2} shows that for finite $N_{\text{B}}$, the outage of ideal
secrecy can be made arbitrarily small by increasing $P_{\text{v}}$.

\begin{example}
Let us apply Theorem \ref{Th2} to the analysis of a USK scheme with $N_{%
\text{A}}=9$, $N_{\text{B}}=4$, $N_{\text{E}}=8$, $\sigma _{\text{E}}^{2}=0$%
, and%
\begin{equation}
P_{\text{v}}=\varepsilon ^{-2/N_{\min }}\kappa (d)^{2}/\Phi ^{2N_{\text{B}%
}/N_{\text{E}}}\text{.}  \label{Pv_e1}
\end{equation}%
We evaluate the secrecy outage probability in (\ref{P_out}) for the $i^{%
\text{th}}$ channel use. We generate $50000$ pairs of mutually independent
complex gaussian random matrices \{$\mathbf{G},\mathbf{H}$\}. For each pair
of \{$\mathbf{G},\mathbf{H}$\}, we evaluate the corresponding realization $%
\tilde{D}$ of the random variable $D$ by%
\begin{equation}
\tilde{D}\approx \frac{\text{vol}(S_{R_{\max }})}{\text{vol}(\Lambda _{%
\mathbb{C}})}=\left( \frac{R_{\max }(P_{\text{v}})}{r_{\text{eff}}(\Lambda _{%
\mathbb{C}})}\right) ^{2N_{\text{B}}}\text{,}
\end{equation}%
where $r_{\text{eff}}(\Lambda _{\mathbb{C}})$ is given in (\ref{r_eff}), $%
R_{\max }(P_{\text{v}})$ is given in (\ref{R_max}). Based on the corresponding $%
50000$ realizations of $D$, we compute the probability of $D<d$, i.e., $P_{%
\text{out}}(d)$. Fig.$~$2 shows the value of $P_{\text{out}}(d)$ as a
function of $\varepsilon $, with $d=2$ and $d=64^{4}$ (large number), respectively. As
expected, the value of $P_{\text{out}}(d)$ decreases with decreasing $%
\varepsilon $, or equivalently, increasing $P_{\text{v}}$.
\end{example}

\begin{figure}[tbp]
\centering\includegraphics[scale=0.65]{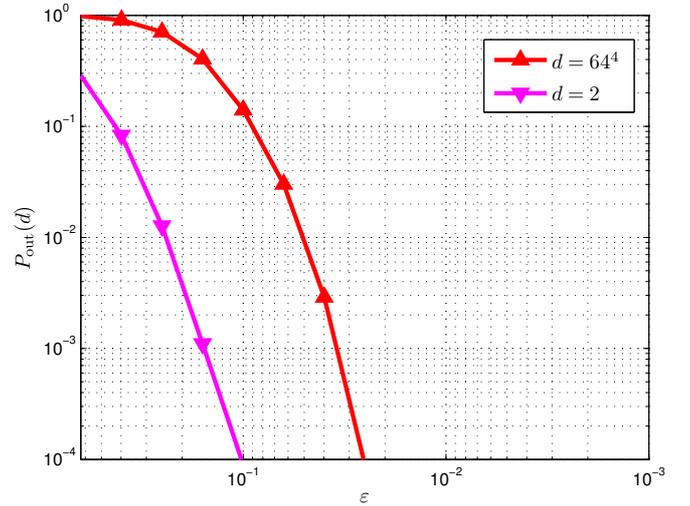}
\caption{$P_{\text{out}}(d)$ vs. $\protect\varepsilon $ with $N_{\text{A}}=9$%
, $N_{\text{B}}=4$ and $N_{\text{E}}=8$. }
\end{figure}

\subsection{Achieving Perfect Secrecy}

From (\ref{SE2}), perfect secrecy requires%
\begin{equation}
H(\mathbf{u|y})=H(\mathbf{u})
\end{equation}%
According to (\ref{uncer}), the problem then reduces to ensuring $%
D_{i}\rightarrow \infty $. From Theorems \ref{Th1} and \ref{Th2}, achieving
perfect secrecy requires infinite AN peak power $P_{\text{v}}$, which is of
theoretical interest only.

\section{Unshared Secret Key Cryptosystem With Finite Constellations}

In this section, we show that the idea of USK can be applied to practical
systems using finite constellations. In this case, we define the concept of
secrecy outage and define a secrecy outage probability. We will later show
how such probability can be made arbitrarily small by considering either
longer blocks of messages or larger constellation size.

\subsection{Encryption}

We consider a sequence of $K$ mutually independent messages $\left\{ \mathbf{%
m}_{l}\right\} _{1}^{K}$, where each one contains $n$ mutually independent
information bits. For each $\mathbf{m}$, Alice maps the corresponding $n$
bits to $N_{\text{B}}$ elements of $\mathbf{{{\mathbf{u}}}}$ for $B$ channel
uses. Each elements of $\mathbf{{{\mathbf{u}}}}$ is uniformly selected from
a $M$-QAM constellation $\mathcal{\tilde{Q}}$, where $\Re (\mathcal{\tilde{Q}%
}\mathbf{)}=\Im (\mathcal{\tilde{Q}}\mathbf{)}=\{0,1,...,\sqrt{M}-1\}$. We
ignore the shifting and scaling operations at Alice to minimize the transmit
power. Consequently, we have%
\begin{equation}
n=BN_{\text{B}}\log _{2}M\text{.}
\end{equation}%
Let $\left\{ \mathbf{{{\mathbf{u}}}}_{i}\right\} _{1}^{B}$ be the block of
transmitted vectors corresponding to one message $\mathbf{m}$.

To secure the total $C=KB$ transmitted vectors $\left\{ \mathbf{u}%
_{j}\right\} _{1}^{C}$, Alice enciphers $\left\{ \mathbf{u}_{j}\right\}
_{1}^{C}$ into the cryptograms $\left\{ \mathbf{y}_{j}\right\} _{1}^{C}$
using a sequence of mutually independent keys $\left\{ \mathbf{v}%
_{j}\right\} _{1}^{C}$. Across the $C$ channel uses, we assume that $\left\{
\mathbf{v}_{j}\right\} _{1}^{C}$ and $\left\{ \mathbf{u}_{j}\right\}
_{1}^{C} $ are mutually independent, and $\left\{ \mathbf{G}_{j}\right\}
_{1}^{C}$ are mutually independent Gaussian random matrices. No assumption
is needed about the statistics of $\left\{ \mathbf{H}_{j}\right\} _{1}^{C}$,
since its realization is known to Alice and Eve.

Since $\left\{ \mathbf{v}_{j}\right\} _{1}^{C}$ and $\left\{ \mathbf{u}%
_{j}\right\} _{1}^{C}$ are mutually independent, using (\ref{SE1}), we only
need to demonstrate the encryption process for one block of transmitted
vectors $\left\{ \mathbf{{{\mathbf{u}}}}_{i}\right\} _{1}^{B}$ corresponding
to a message $\mathbf{m}$.

The encryption process is the same as that of the infinite constellation
case: for the $i^{\text{th}}$ channel use, Alice independently chooses a one
time pad key $\mathbf{v}_{i}$ from the set $S$ in (\ref{set}), and encrypts $%
\mathbf{{{\mathbf{u}}}}_{i}$ to $\mathbf{y}_{i}$ in (\ref{Eve_U}) using $%
\mathbf{v}_{i}$, such that $\mathbf{G}_{i}\mathbf{\mathbf{\mathbf{\mathbf{%
\mathbf{V}}}}}_{1,i}\mathbf{\mathbf{{{\mathbf{u}}}}}_{i}$ is the $k_{i}^{%
\text{th}}$ closest lattice point to $\mathbf{y}_{i}$, within the infinite
lattice%
\begin{equation}
\Lambda _{\mathbb{C},i}=\{\mathbf{G}_{i}\mathbf{\mathbf{\mathbf{\mathbf{%
\mathbf{V}}}}}_{1,i}\mathbf{\mathbf{{{\mathbf{u}}}}},\mathbf{u}\in {\mathbb{Z%
}\left[ i\right] ^{N_{\text{B}}}\}}\text{.}
\end{equation}%
The value of $k_{i}$ ranges from $1$ to $D_{i}$, where%
\begin{equation}
D_{i}=|S_{R_{\max ,i}}\cap \Lambda _{\mathbb{C},i}|\text{,}  \label{Di}
\end{equation}%
and $S_{R_{\max ,i}}$ is a sphere centered at $\mathbf{y}_{i}$ with radius:%
\begin{equation}
R_{\max ,i}(P_{\text{v}})\triangleq \max_{||\mathbf{v}_{i}\mathbf{||}%
^{2}\leq P_{\text{v}}}\left\Vert \mathbf{G}_{i}\mathbf{Z}_{i}\mathbf{v}%
_{i}\right\Vert =\sqrt{\lambda _{\max ,i}P_{\text{v}}}\text{.}
\label{R_maxi}
\end{equation}%
where $\lambda _{\max ,i}$ is the largest eigenvalue of $(\mathbf{G}_{i}%
\mathbf{Z}_{i})^{H}(\mathbf{G}_{i}\mathbf{Z}_{i})$. As shown in Fig.~3, $%
D_{i}$ represents the total number of points within the sphere $S_{R_{\max
,i}}$.

Different from the infinite constellation case, the condition $D_{i}\geq 2$
in (\ref{uncer}) cannot ensure $H(\mathbf{{{\mathbf{u}}}}_{i}\mathbf{|y}%
_{i})>0$. The reason is that Eve knows that $\mathbf{G}_{i}\mathbf{\mathbf{%
\mathbf{\mathbf{\mathbf{V}}}}}_{1,i}\mathbf{u}_{i}$ is a finite lattice
constellation, i.e., a finite subset of $\Lambda _{\mathbb{C},i}$:%
\begin{equation}
\Lambda _{\text{F},i}\triangleq \{\mathbf{G}_{i}\mathbf{\mathbf{\mathbf{%
\mathbf{\mathbf{V}}}}}_{1,i}\mathbf{\mathbf{{{\mathbf{u}}}}},{\mathbf{{%
\mathbf{u}}}}\in \mathcal{\tilde{Q}}{^{N_{\text{B}}}\}}\text{.}
\label{FN_LT}
\end{equation}

Since $k_{i}$ is a function of $\mathbf{v}_{i}$, which is randomly and
independently selected by Alice and is never disclosed to anyone, Eve can
neither know the distribution of $k_{i}$. Given $\mathbf{y}_{i}$, Eve only
knows that $\mathbf{G}_{i}\mathbf{\mathbf{\mathbf{\mathbf{\mathbf{V}}}}}%
_{1,i}\mathbf{\mathbf{{{\mathbf{u}}}}}_{i}\in S_{R_{\max ,i}}\cap \Lambda _{%
\text{F},i}$. Let $L_{i}$ be the cardinality of such choices, i.e.,%
\begin{equation}
L_{i}=|S_{R_{\max ,i}}\cap \Lambda _{\text{F},i}|\text{.}  \label{LII}
\end{equation}%
Since $\Lambda _{\text{F},i}\subset \Lambda _{\mathbb{C},i}$, we have%
\begin{equation}
1\leq L_{i}\leq D_{i}.
\end{equation}%
As shown in Fig.~3, $L_{i}$ represents the number of solid points within the
sphere $S_{R_{\max ,i}}$.

\begin{remark}
Due to the use of finite constellation $\mathcal{\tilde{Q}}^{N_{\text{B}}}$,
we redefine the effective secrecy key $k_{i}$ as $k_{\text{F},i}$, that is, $%
\mathbf{G}_{i}\mathbf{\mathbf{\mathbf{\mathbf{\mathbf{V}}}}}_{1,i}\mathbf{%
\mathbf{{{\mathbf{u}}}}}_{i}$ is the $k_{\text{F},i}^{\text{th}}$ closest
lattice point to $\mathbf{y}_{i}$, within the finite lattice constellation $%
\Lambda _{\text{F},i}$. The corresponding key space size is $L_{i}$ per
channel use.
\end{remark}

\begin{remark}
The practical secrecy scheme \cite{Liu13Letter} is a special case of USK
cryptosystem with $k_{\text{F},i}\geq 2$.
\end{remark}

\subsection{Analyzing Eve's Equivocation}

We then show that Eve's equivocation $H(\left\{ \mathbf{{{\mathbf{u}}}}%
_{i}\right\} _{1}^{B}|\left\{ \mathbf{y}_{i}\right\} _{1}^{B})$ is
determined by $\left\{ L_{i}\right\} _{1}^{B}$. The posterior probability
that Eve obtains $\mathbf{{{\mathbf{u}}}}_{i}$, or equivalently, finds $k_{%
\text{F},i}$, is equal to%
\begin{equation}
\Pr \left\{ \mathbf{{{\mathbf{u}}}}_{i}\mathbf{|y}_{i}\right\} =\Pr \left\{
k_{\text{F},i}\mathbf{|y}_{i}\right\} =\Pr \left\{ \mathbf{{{\mathbf{u}}}}%
_{i}\mathbf{|{{\mathbf{u}}}}_{i}\in \mathcal{U}_{\text{F},i}\right\} \text{.}
\label{PPP_FI}
\end{equation}%
where%
\begin{equation}
\mathcal{U}_{\text{F},i}\triangleq \left\{ \mathbf{u}^{\prime }\text{: }%
\mathbf{G}_{i}\mathbf{\mathbf{\mathbf{\mathbf{\mathbf{V}}}}}_{1,i}\mathbf{%
\mathbf{{{\mathbf{u}}}}}^{\prime }\in S_{R_{\max ,i}}\cap \Lambda _{\text{F}%
,i}\right\} \text{.}  \label{UF}
\end{equation}

Due to the use of uniform constellation $\mathcal{\tilde{Q}}^{N_{\text{B}}}$%
, according to Bayes' theorem, we have%
\begin{equation}
\Pr \left\{ \mathbf{{{\mathbf{u}}}}_{i}\mathbf{|{{\mathbf{u}}}}_{i}\in
\mathcal{U}_{\text{F},i}\right\} =\frac{1}{L_{i}}\text{.}  \label{Post_FI}
\end{equation}

To recover one message $\mathbf{m}$, Eve has to recover all vectors in $%
\left\{ \mathbf{{{\mathbf{u}}}}_{i}\right\} _{1}^{B}$, or equivalently, find
$\left\{ k_{\text{F},i}\right\} _{1}^{B}$. Therefore, Eve's equivocation is
given by%
\begin{equation}
H(\mathbf{m|}\left\{ \mathbf{y}_{i}\right\} _{1}^{B})=H(\left\{ k_{\text{F}%
,i}\right\} _{1}^{B}\mathbf{|}\left\{ \mathbf{y}_{i}\right\}
_{1}^{B})=H(\left\{ \mathbf{{{\mathbf{u}}}}_{i}\right\} _{1}^{B}|\left\{
\mathbf{y}_{i}\right\} _{1}^{B})\text{.}  \label{uncer_FI1}
\end{equation}%
Moreover, since $\mathbf{u}_{i}$ is independent of $\mathbf{u}_{j}$ and $%
\mathbf{y}_{j}$, we have%
\begin{equation}
H(\left\{ \mathbf{{{\mathbf{u}}}}_{i}\right\} _{1}^{B}|\left\{ \mathbf{y}%
_{i}\right\} _{1}^{B})=\sum_{i=1}^{B}H(\mathbf{u}_{i}\mathbf{|y}%
_{i})=\sum_{i=1}^{B}\log L_{i}\text{.}  \label{uncer_FI}
\end{equation}%
\begin{figure}[tbp]
\centering\includegraphics[scale=0.65]{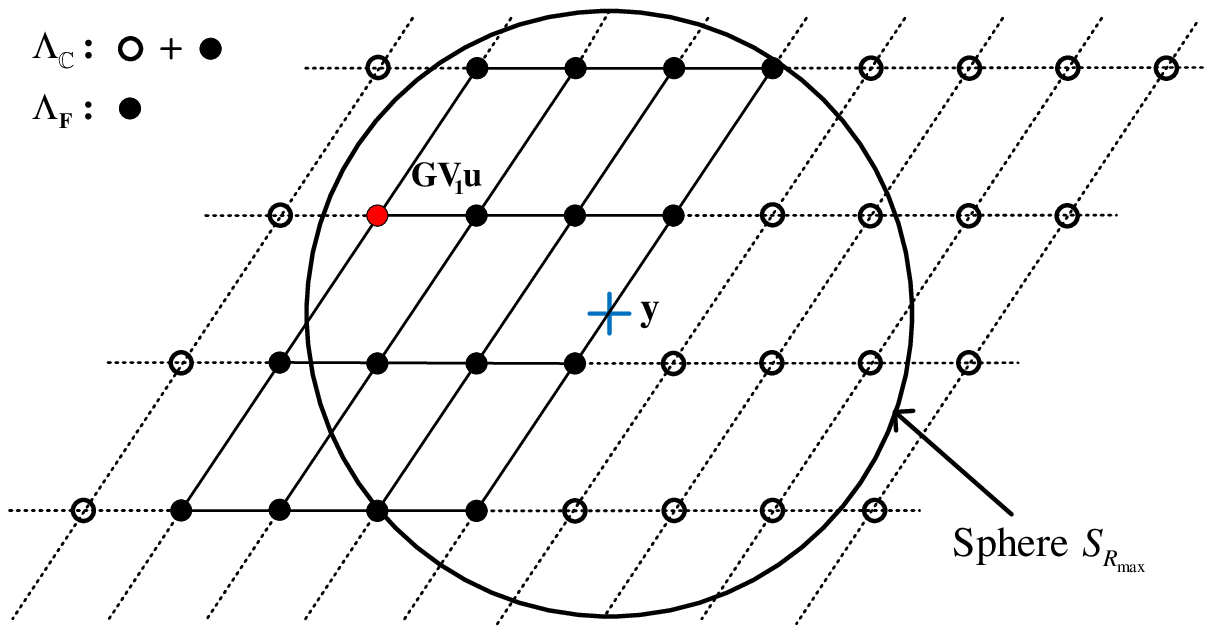}
\caption{ The USK cryptosystem with finite constellations.}
\end{figure}

\subsection{Ideal Secrecy Outage}

Based on (\ref{uncer_FI}), Eve's equivocation is dominated by the values in $%
\left\{ L_{i}\right\} _{1}^{B}$, which are known to Eve. From Alice's
perspective, according to (\ref{LII}) and (\ref{UF}), $L_{i}$ is a function
of $\mathbf{G}_{i}$, thus a random variable. Although Alice cannot know the
exact values in $\left\{ L_{i}\right\} _{1}^{B}$, she may be able to
evaluate the cdf of Eve's equivocation, given by%
\begin{eqnarray}
\Pr \left\{ \sum_{i=1}^{B}\log L_{i}<\log d\right\} &\leq &\Pr \left\{ \log
L_{i}<\log d,1\leq i\leq B\right\}  \notag \\
&=&\Pr \left\{ L_{1}<d\text{, ..., }L_{B}<d\right\}  \notag \\
&\triangleq &P_{\text{F,out}}(d,B)\text{.}  \label{PF_out}
\end{eqnarray}%
where $2\leq d\leq M^{N_{\text{B}}}$.

We refer to the event
\begin{equation}
\sum_{i=1}^{B}\log L_{i}<\log d\text{,}
\end{equation}%
as the \emph{secrecy outage }due to the use of the finite constellation $%
\mathcal{\tilde{Q}}^{N_{\text{B}}}$. We refer to $P_{\text{F,out}}(d,B)$ as
the secrecy outage probability. From (\ref{uncer_FI}) and (\ref{PF_out}), if
$P_{\text{F,out}}(d,B)\rightarrow 0$,%
\begin{equation}
H(\left\{ \mathbf{{{\mathbf{u}}}}_{i}\right\} _{1}^{B}|\left\{ \mathbf{y}%
_{i}\right\} _{1}^{B})=H(\left\{ k_{\text{F},i}\right\} _{1}^{B}\mathbf{|}%
\left\{ \mathbf{y}_{i}\right\} _{1}^{B})\geq \log d\text{.}  \label{target}
\end{equation}

In the next section, we will show that Alice can ensure $P_{\text{F,out}%
}(d,B)\rightarrow 0$ by increasing the message block size $B$ with certain $%
M $ and $P_{\text{v}}$.

\section{The Security of USK with Finite Constellations}

In this section, we show that the USK with the finite constellation $%
\mathcal{\tilde{Q}}^{N_{\text{B}}}$ provides Shannon's ideal secrecy with an
arbitrarily small outage. To prove the main theorems, we first introduce the
following lemma.

We define%
\begin{equation}
\Theta (P_{\text{v}})\triangleq \frac{2R_{\max }(P_{\text{v}})}{\sqrt{M}r_{%
\text{eff}}(\Lambda _{\mathbb{C}})}\text{.}  \label{RRR}
\end{equation}%
where $r_{\text{eff}}(\Lambda _{\mathbb{C}})$ is given in (\ref{r_eff}) and $%
R_{\max }(P_{\text{v}})$ is given in (\ref{R_maxi}). From Alice perspective,
$\Theta (P_{\text{v}})$ is a function of $\mathbf{G}$, thus is a random
variable. Its cdf is bounded by the following lemma.

\begin{lemma}
\label{Th4}%
\begin{align}
& \Pr \left\{ \Theta (P_{\text{v}})<x\right\}  \notag \\
& >\prod\limits_{j=1}^{N_{\text{B}}}B_{\substack{ N_{\text{E}}\left( N_{%
\text{A}}-N_{\text{B}}\right) ,  \\ N_{\text{E}}-j+1}}\left( \frac{N_{\text{E%
}}\left( N_{\text{A}}-N_{\text{B}}\right) g(x,j)}{N_{\text{E}}\left( N_{%
\text{A}}-N_{\text{B}}\right) g(x,j)+N_{\text{E}}-j+1}\right) \text{,}
\label{P_t}
\end{align}%
where%
\begin{equation}
g(x,j)=\frac{x^{2}MN_{\text{B}}(N_{\text{E}}-j+1)}{4\pi eP_{\text{v}}N_{%
\text{E}}\left( N_{\text{A}}-N_{\text{B}}\right) }\text{,}  \label{g_x}
\end{equation}%
and $B_{a\text{,}b}(x)$ is the \emph{regularized incomplete beta} function
\cite{Imbeta}:%
\begin{equation}
B_{a\text{,}b}(x)\triangleq \sum_{j=a}^{a+b-1}\left(
\begin{array}{c}
a+b-1 \\
j%
\end{array}%
\right) x^{j}(1-x)^{a+b-1-j}\text{.}  \label{Beta}
\end{equation}
\end{lemma}

\begin{proof}
See Appendix E.
\end{proof}

\subsection{Achieving Ideal Secrecy}

As shown in (\ref{SE1}) and\ (\ref{uncer_FI}), ideal secrecy is achieved
when $\sum_{i=1}^{B}\log L_{i}>0$. From (\ref{PF_out}), the problem then
reduces to ensuring%
\begin{equation}
P_{\text{F,out}}(d,B)\rightarrow 0\text{,}  \label{Li}
\end{equation}%
for any $d\geq 2$. Lemma \ref{Th4} enables us to prove the following theorem.

\begin{theo}
\label{Th3}If $\varepsilon <1$, $d\geq 2$, $P_{\text{v}}=\varepsilon
^{-2/N_{\min }}\kappa (d)^{2}/\Phi ^{2N_{\text{B}}/N_{\text{E}}}$, and $%
M\geq \varepsilon ^{-3-2/N_{\min }}\kappa (d)^{2}$, then%
\begin{equation}
P_{\text{F,out}}(d,B)<O(\varepsilon ^{B})\text{,}
\end{equation}%
where $\kappa (d)$ is given in (\ref{kappa}) and $\Phi $ is given in (\ref%
{Pi_SVD}), i.e., ideal secrecy is achieved with probability $1-O(\varepsilon
^{B})$.
\end{theo}

\begin{proof}
See Appendix F.
\end{proof}

Theorem \ref{Th3} shows that for finite $N_{\text{B}}$ and finite
constellation $\mathcal{\tilde{Q}}^{N_{\text{B}}}$, the ideal secrecy outage
can be made arbitrarily small. Given a desired pair $\left\{ \varepsilon
,d\right\} $, Alice can easily compute the required values of $P_{\text{v}}$
and $M$ to realize the USK cryptosystem.

\begin{example}
We consider a USK scheme with $N_{\text{A}}=4$, $N_{\text{B}}=2$, $N_{\text{E%
}}=3$ and $\sigma _{\text{E}}^{2}=0$. To apply Theorem \ref{Th3}, we fix $d=2
$ and consider two cases where $\varepsilon =0.3981$ and $0.1990$. The
conditions in Theorem \ref{Th3} then reduce to%
\begin{eqnarray}
P_{\text{v}} &=&1.8306\text{ and }M\geq 15.9659\text{, }\ \text{for }%
\varepsilon =0.3981\text{,}  \notag \\
P_{\text{v}} &=&3.6620\text{ and }M\geq 255.7297\text{, for }\varepsilon
=0.1990\text{.}
\end{eqnarray}%
Fig.~4 compares the value of $P_{\text{F,out}}(2,B)$ as a function of $B$.
Note that $P_{\text{F,out}}(2,B)$ can be written as%
\begin{equation}
\Pr \left\{ L_{1}=1\text{, ..., }L_{B}=1\right\} =\Pr \left\{
\sum_{i=1}^{B}\log L_{i}=0\right\} \text{.}
\end{equation}%
We observe that $P_{\text{F,out}}(2,B)=4.6250\times 10^{-4}$ when $P_{\text{v%
}}=3.6620$, $M=256$, and $B=1$. It confirms that the secrecy outage
probability can be made arbitrarily small by increasing $P_{\text{v}}$ and $M
$. Meanwhile, we observe that the secrecy outage probability decreases
exponentially with $B$.
\end{example}

\begin{remark}
For the finite constellation case, the value of target equivocation at Eve
is given by $\log d$ in (\ref{target}). Note that this is not easily
computable for the infinite constellation case according to (\ref{uncer}).
\end{remark}

\begin{figure}[tbp]
\centering\includegraphics[scale=0.65]{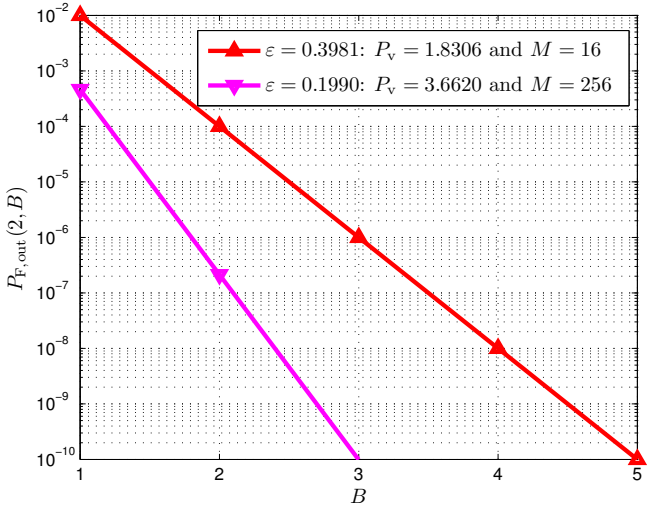}
\caption{$P_{\text{F,out}}(2,B)$ vs. $M$ and $B$ with $N_{\text{A}}=4$, $N_{%
\text{B}}=2$ and $N_{\text{E}}=3$. }
\end{figure}

\subsection{Peak AN-to-Signal Power Ratio}

By shifting and scaling, $\mathbf{u}\in \mathcal{\tilde{Q}}^{N_{\text{B}}}$
can be converted into a regular $M$-QAM symbol $\mathbf{\bar{u}}\in \mathcal{%
Q}^{N_{\text{B}}}$. To measure the power efficiency of the proposed USK
cryptosystem, we define%
\begin{equation}
r\triangleq \frac{P_{\text{v}}}{\text{E}(||\mathbf{V}_{1}\mathbf{\bar{u}||}%
^{2})}\text{,}  \label{r_def}
\end{equation}%
as the ratio of the peak AN power $P_{\text{v}}$ and the average transmitted
signal power.

Since%
\begin{equation}
\text{E}(||\mathbf{V}_{1}\mathbf{\bar{u}||}^{2})=\text{E}(||\mathbf{\bar{u}||%
}^{2})=\frac{2\left( M-1\right) N_{\text{B}}}{3}\text{,}
\end{equation}%
the corresponding ratio as a function of $P_{\text{v}}$ is given by%
\begin{equation}
r=\frac{3P_{\text{v}}}{2\left( M-1\right) N_{\text{B}}}\text{.}
\label{P_ratio}
\end{equation}

\begin{example}
Under the same setting of Example 2, if $M=256$, $r=1.08\%$. We see
that the proposed USK cryptosystem is very practical, since it
requires a very small proportion of the total transmission power.
Note that the value of $r$ can be further reduced by increasing the
constellation size $M$.
\end{example}

\section{Discussions}

\subsection{USK Cryptosystems vs. Previous AN based Schemes}

The existing AN based security schemes \cite{Goel08,Zhou10,XZhang13}
leverage infinite-length wiretap codes, where the aim is to achieve strong
secrecy.

In contrast, the proposed USK cryptosystem is valid for any coded/uncoded
MIMO with finite block length and QAM signalling. Our scheme achieves
Shannon's ideal secrecy with an arbitrarily small outage probability.

\subsection{Extension to the Case of $N_{\text{E}}\geq N_{\text{A}}$}

The constraint $N_{\text{E}}<N_{\text{A}}$ is a common assumption that
appears in the vast literature on AN based schemes \cite%
{Goel08,Zhou10,XZhang13}. Under this condition, we have shown the existence
of an unshared secret key cryptosystem which provides Shannon's ideal
secrecy.

If $N_{\text{E}}\geq N_{\text{A}}$, $\mathbf{G}$ has a left inverse, denoted
by $\mathbf{G}^{\dag }$, then Eve can remove the unshared secret key $%
\mathbf{v}$ by multiplying $\mathbf{y}$ by $\mathbf{{{\mathbf{W}}}=HG{^{\dag
}}}$, i.e.,%
\begin{equation}
\mathbf{\mathbf{W}\mathbf{y=}H\mathbf{V}}_{1}\mathbf{{{\mathbf{u}}}}+\mathbf{%
{{\mathbf{W}}}n}_{\text{E}}\text{.}
\end{equation}

We can show that this attack amplifies Eve's channel noise greatly.
Consequently, $\mathbf{n}_{\text{E}}$ takes the role of the unshared secret
key. We can show that with certain amount of $\sigma _{\text{E}}^{2}$, ideal
secrecy is achievable. This result will be reported in our next paper.

\section{Conclusions}

We have exploited the role that artificial noise plays in physical layer
security to show that it can be used as an unshared one-time pad secret key.
The proposed unshared secret key (USK) cryptosystem with an infinite lattice
input alphabet provides Shannon's ideal secrecy and perfect secrecy by
tuning the power allocated to the artificial noise component. Moreover,
unlike the traditional AN technique, this USK system can be applied to
practical systems using finite lattice constellations. We have shown that
ideal secrecy can be obtained with an arbitrarily small outage probability.
Our results provide analytical insights relating cryptography and physical
layer security on a fundamental level. Future work will generalize USK to
relaying networks.

\section*{Acknowledgment}

We thank Professor Terence Tao for his constructive comments and for
suggesting references, and Professor Joseph Jean Boutros for his helpful
discussions. We thank Professor Ram Zamir for pointing out the boundary
effect of finite constellations.

\section*{Appendix}

\subsection{Proof of Lemma 1}

Recalling that%
\begin{equation}
\Delta (d)=\frac{\kappa (d)^{2N_{\text{E}}}|\det ((\mathbf{GV}_{1}\mathbf{)}%
^{H}(\mathbf{GV}_{1}\mathbf{))|}}{P_{\text{v}}^{N_{\text{E}}}}\text{.}
\label{Pc_upp}
\end{equation}

From Alice's perspective, $\mathbf{G}$ is a complex Gaussian random matrix.
The matrix $\mathbf{V}_{1}$ with orthonormal columns is known. According to
\cite{Lukacs54}, $\mathbf{GV}_{1}$ a Gaussian random matrix with i.i.d.
elements. Moreover, $|\det ((\mathbf{GV}_{1}\mathbf{)}^{H}(\mathbf{GV}_{1}%
\mathbf{))|}$ can be expressed as the product of independent Chi-squared
variables \cite{Verdu04}:%
\begin{equation}
|\det ((\mathbf{GV}_{1}\mathbf{)}^{H}(\mathbf{GV}_{1}\mathbf{))|}%
=\prod_{i=1}^{N_{\text{B}}}\frac{1}{2}\mathcal{X}^{2}\left( 2(N_{\text{E}%
}-i+1)\right) \text{.}  \label{r1}
\end{equation}

Using the properties of the Chi-squared distribution and central limit
theorem, as $N_{\text{B}}\rightarrow \infty $, we have%
\begin{equation}
\frac{\sum\limits_{i=1}^{N_{\text{B}}}\log \mathcal{X}^{2}\left( 2(N_{\text{E%
}}-i+1)\right) -A}{\sqrt{V}}\overset{a.s.}{\rightarrow }\mathcal{N}(0,1)%
\text{,}  \label{clt}
\end{equation}%
where%
\begin{equation*}
A=\sum\limits_{i=1}^{N_{\text{B}}}\text{E}\left( \log \mathcal{X}^{2}\left(
2(N_{\text{E}}-i+1)\right) \right) \text{,}
\end{equation*}%
\begin{equation*}
V=\sum\limits_{i=1}^{N_{\text{B}}}\text{Var}\left( \log \mathcal{X}%
^{2}\left( 2(N_{\text{E}}-i+1)\right) \right) \text{.}
\end{equation*}

Using the properties of Log Chi-squared distributions \cite{Lee12}, we have%
\begin{equation*}
A=\sum\limits_{k=N_{\text{E}}-N_{\text{B}}+1}^{N_{\text{E}}}\left( \log
2+\psi (k)\right) \text{,}
\end{equation*}%
\begin{equation*}
V=\sum\limits_{k=N_{\text{E}}-N_{\text{B}}+1}^{N_{\text{E}}}\psi _{1}(k)%
\text{,}
\end{equation*}%
where $\psi (x)=\frac{d}{dx}\log \Gamma (x)$ is the \emph{digamma} function,
and $\psi _{1}(x)=\frac{d^{2}}{dx^{2}}\log \Gamma (x)$ is the \emph{trigamma}
function.

Informally, we may write (\ref{r1}) and (\ref{clt}) as%
\begin{equation}
|\det ((\mathbf{GV}_{1}\mathbf{)}^{H}(\mathbf{GV}_{1}\mathbf{))|}\approx
2^{-N_{\text{B}}}e^{A+\mathcal{N}(0,\text{ }V)}\text{.}  \label{33}
\end{equation}

According to (\ref{33}), as $N_{\text{B}}\rightarrow \infty $, $\Delta (d)$
converges to the random variable $\Omega $:%
\begin{equation}
\Omega \triangleq \frac{\kappa (d)^{2N_{\text{E}}}\exp \left( A+\mathcal{N}%
\left( 0,V\right) \right) }{2^{N_{\text{B}}}P_{\text{v}}^{N_{\text{E}}}}%
\text{.}  \label{Omega_g}
\end{equation}

To simplify the expressions of $A$ and $V$, we use the following
approximations \cite{Lee12}:%
\begin{eqnarray}
\psi (k) &\approx &\log k-1/(2k)\text{,}  \notag \\
\psi _{1}(k) &\approx &1/k\text{.}  \label{psi_def}
\end{eqnarray}%
Then, we have%
\begin{equation}
V\leq \sum\limits_{i=1}^{N_{\text{B}}}\frac{1}{k}\leq \log N_{\text{B}%
}+\varsigma <\log 2N_{\text{B}}\text{,}  \label{V1V2}
\end{equation}%
where $\varsigma $ is Euler--Mascheroni constant. Similarly, we have%
\begin{eqnarray}
A &=&\sum\limits_{k=N_{\text{E}}-N_{\text{B}}+1}^{N_{\text{E}}}\left( \log
2+\log k-\frac{1}{2k}\right)  \notag \\
&<&N_{\text{B}}\log 2+\log \Phi ^{-2N_{\text{B}}}\text{,}  \label{AA}
\end{eqnarray}%
where%
\begin{equation}
\Phi =\left[ \frac{(N_{\text{E}}-N_{\text{B}})!}{N_{\text{E}}!}\right] ^{%
\frac{1}{2N_{\text{B}}}}\text{.}
\end{equation}

From (\ref{AA}) and (\ref{Omega_g}), $\Omega $ can be upper bounded by%
\begin{equation}
\Omega <\frac{\kappa (d)^{2N_{\text{E}}}\exp \left( \mathcal{N}\left(
0,V\right) \right) }{\Phi ^{2N_{\text{B}}}P_{\text{v}}^{N_{\text{E}}}}\text{.%
}  \label{Omeg1}
\end{equation}

Recall that $N_{\text{E}}\geq N_{\text{B}}$. By substituting $P_{\text{v}%
}\geq \rho ^{2}/\Phi ^{2N_{\text{B}}/N_{\text{E}}}$ and $\rho >\kappa (d)$
to the right side of (\ref{Omeg1}), we have%
\begin{equation}
\Omega <\frac{\exp \left( \mathcal{N}\left( 0,V\right) \right) }{\left( \rho
/\kappa (d)\right) ^{2N_{\text{E}}}}\leq \frac{\exp \left( \mathcal{N}\left(
0,V\right) \right) }{\left( \rho /\kappa (d)\right) ^{2N_{\text{B}}}}%
\triangleq \Omega _{\text{UB}}\text{,}
\end{equation}%
and%
\begin{align}
& \text{ \ \ }\Pr \left\{ \Delta (d)>\left( \rho /\kappa (d)\right) ^{-N_{%
\text{B}}}\right\}  \notag \\
& <\Pr \left\{ \Omega _{\text{UB}}>\left( \rho /\kappa (d)\right) ^{-N_{%
\text{B}}}\right\}  \notag \\
& =\Pr \left\{ \mathcal{N}\left( 0,V\right) >N_{\text{B}}\log \left( \rho
/\kappa (d)\right) \right\}  \notag \\
& <1/2\exp \left( -\frac{N_{\text{B}}^{2}\log ^{2}\left( \rho /\kappa
(d)\right) }{2V}\right)  \notag \\
& \overset{a}{<}1/2\exp \left( -\frac{N_{\text{B}}^{2}\log ^{2}\left( \rho
/\kappa (d)\right) }{2\log 2N_{\text{B}}}\right)  \notag \\
& =O\left( \left( \rho /\kappa (d)\right) ^{-N_{\text{B}}}\right) \text{,}
\label{Gauss_upp}
\end{align}%
where ($a$) holds because of (\ref{V1V2}).

From (\ref{Gauss_upp}) and (\ref{Pc_upp}), if $\rho >\kappa (d)$, as $N_{%
\text{B}}\rightarrow \infty $, we have $\Delta (d)\overset{a.s.}{\rightarrow
}0$.\QEDA
\vspace{-5 mm}

\subsection{Proof of Lemma 2}

We recall (\ref{Pc_upp}) and (\ref{r1}) and consider the random variable%
\begin{equation}
\Psi \triangleq \prod_{i=1}^{N_{\text{B}}}\frac{\mathcal{X}^{2}\left( 2(N_{%
\text{E}}-i+1)\right) }{2(N_{\text{E}}-i+1)}\text{.}
\end{equation}%
\

Recalling that $N_{\text{E}}\geq N_{\text{B}}$. By substituting $\Psi $, $P_{%
\text{v}}\geq \rho ^{2}/\Phi ^{2N_{\text{B}}/N_{\text{E}}}$ and $\rho
>\kappa (d)$ to the right side of (\ref{Pc_upp}), we have%
\begin{eqnarray}
\Delta (d) &=&\left( \rho /\kappa (d)\right) ^{-2N_{\text{E}}}\Psi  \notag \\
&\leq &\left( \rho /\kappa (d)\right) ^{-2N_{\text{B}}}\Psi \text{.}
\end{eqnarray}

Consequently, we obtain%
\begin{align}
& \text{ \ \ }\Pr \left\{ \Delta (d)>\left( \rho /\kappa (d)\right) ^{-N_{%
\text{B}}}\right\}  \notag \\
& \leq \Pr \left\{ \Psi \left( \rho /\kappa (d)\right) ^{-2N_{\text{B}%
}}>\left( \rho /\kappa (d)\right) ^{-N_{\text{B}}}\right\}  \notag \\
& =\Pr \left\{ \Psi >(\rho /\kappa (d))^{N_{\text{B}}}\right\}  \notag \\
& \overset{a}{\leq }\Pr \left\{ \sum_{i=1}^{N_{\text{B}}}\frac{\mathcal{X}%
^{2}\left( 2(N_{\text{E}}-i+1)\right) }{2(N_{\text{E}}-i+1)}>N_{\text{B}%
}\rho /\kappa (d)\right\}  \notag \\
& <\sum_{i=1}^{N_{\text{B}}}\Pr \left\{ \mathcal{X}^{2}\left( 2(N_{\text{E}%
}-i+1)\right) \geq 2(N_{\text{E}}-i+1)\rho /\kappa (d)\right\}  \notag \\
& \leq \sum_{i=1}^{N_{\text{B}}}\left( e^{1-\rho /\kappa (d)}\rho /\kappa
(d)\right) ^{N_{\text{E}}-i+1}  \notag \\
& \triangleq \Upsilon (\rho /\kappa (d))\text{,}
\end{align}%
where ($a$) holds due to the inequality of arithmetic and geometric means.%
\QEDA\vspace{-5mm}

\subsection{Proof of Lemma 3}

We pick an element $\mathbf{v}_{0}$ from $\mathcal{S}$ with $||\mathbf{v}%
_{0}||^{2}=P_{\text{v}}$. Suppose that $\mathbf{v}_{0}\in S_{k_{0}}$, where $%
k_{0}$ is the corresponding effective secret key. Since $D\geq k_{0}$, we
have%
\begin{equation}
F_{D}(d,P_{\text{v}})=\Pr \left\{ D<d\right\} <\Pr \left\{ k_{0}\leq
d\right\} \text{.}  \label{l1}
\end{equation}%
The problem then reduces to evaluating $\Pr \left\{ k_{0}\leq d\right\} $.

Let $\mathcal{S}_{R}$ be a sphere of radius $R\leq R_{\max }(P_{\text{v}})$
centered at $\mathbf{y}$, where vol$(\mathcal{S}_{R})=d\cdot $vol$(\Lambda _{%
\mathbb{C}})$ (see Fig.~\ref{fig_k_th_point}). Let $K$ be the number of the
points in $\mathcal{S}_{R}\cap \Lambda _{\mathbb{C}}$. We have%
\begin{equation}
K\approx \frac{\text{vol}(\mathcal{S}_{R})}{\text{vol}(\Lambda _{\mathbb{C}})%
}=d\text{.}  \label{K_d}
\end{equation}

If $\mathbf{GV}_{1}\mathbf{u}\in \mathcal{S}_{R}$, we have $k_{0}\leq d$,
and vice versa. Thus, the two events are equivalent, i.e.,%
\begin{equation}
\Pr \left\{ k_{0}\leq d\right\} =\Pr \left\{ \mathbf{GV}_{1}\mathbf{u}\in
\mathcal{S}_{R}\right\} \text{.}  \label{l2}
\end{equation}

Let $\mathcal{S}_{R}^{\prime }$ be a sphere with the same radius $R$
centered at $\mathbf{GV}_{1}\mathbf{u}$. If $\mathbf{GV}_{1}\mathbf{u}\in
\mathcal{S}_{R}$, then $\mathbf{y}\in \mathcal{S}_{R}^{\prime }$, and vice
versa. Thus, the two events are equivalent, i.e.,%
\begin{equation}
\Pr \left\{ \mathbf{GV}_{1}\mathbf{u}\in \mathcal{S}_{R}\right\} =\Pr
\left\{ \mathbf{y}\in \mathcal{S}_{R}^{\prime }\right\} \text{.}  \label{l3}
\end{equation}

From (\ref{l1}), (\ref{l2}) and (\ref{l3}), we have%
\begin{eqnarray}
&&F_{D}(d,P_{\text{v}})  \notag \\
&<&\Pr \left\{ \mathbf{y}\in \mathcal{S}_{R}^{\prime }\right\}  \notag \\
&=&\Pr \left\{ \mathbf{y}\in \mathcal{S}_{R}^{\prime }|\text{vol}(\mathcal{S}%
_{R}^{\prime })\leq C\right\} \cdot \Pr \left\{ \text{vol}(\mathcal{S}%
_{R}^{\prime })\leq C\right\} +  \notag \\
&&\Pr \left\{ \mathbf{y}\in \mathcal{S}_{R}^{\prime }|\text{vol}(\mathcal{S}%
_{R}^{\prime })>C\right\} \cdot \Pr \left\{ \text{vol}(\mathcal{S}%
_{R}^{\prime })>C\right\}  \notag \\
&<&\Pr \left\{ \mathbf{y}\in \mathcal{S}_{R}^{\prime }|\text{vol}(\mathcal{S}%
_{R}^{\prime })\leq C\right\} +\Pr \left\{ \text{vol}(\mathcal{S}%
_{R}^{\prime })>C\right\}\label{Pd}
\end{eqnarray}%
where $C$ is a positive number.

We then evaluate the two terms in (\ref{Pd}) separately. We use the same
settings as Lemmas \ref{lem2} and \ref{lem3}, i.e., $P_{\text{v}}\geq \rho
^{2}/\Phi ^{2N_{\text{B}}/N_{\text{E}}}$, $\rho >\kappa (d)$. We set%
\begin{equation}
C=\pi ^{N_{\text{E}}}P_{\text{v}}^{N_{\text{E}}}\left( \frac{\rho }{\kappa
(d)}\right) ^{-N_{\text{B}}}\text{.}
\end{equation}

\emph{1) }$\Pr \left\{ \mathbf{y}\in \mathcal{S}_{R}^{\prime }|\text{vol}(%
\mathcal{S}_{R}^{\prime })\leq C\right\} $\emph{:} Let $\mathcal{S}_{\text{C}%
}$ be a sphere centered at $\mathbf{GV}_{1}\mathbf{u}$, where vol$(\mathcal{S%
}_{\text{C}})=C$. Let $\mathcal{S}_{\text{C0}}$ be a sphere centered at the
origin, where vol$(\mathcal{S}_{\text{C0}})=C$. Recalling that Alice knows $%
\mathbf{Z}$ and $\mathbf{v}_{0}$. For $\mathbf{G}$, Alice knows its
statistics, but doesn't know its realization. Therefore, from Alice
perspective, $\mathbf{\tilde{n}}_{\text{v}}=\mathbf{G}\mathbf{Zv}_{0}$ has
i.i.d. $\mathcal{N}_{\mathbb{C}}(0$, $P_{\text{v}})$ components \cite%
{Lukacs54}.

Therefore, we have
\begin{eqnarray}
&&\Pr \left\{ \mathbf{y}\in \mathcal{S}_{R}^{\prime }|\text{vol}(\mathcal{S}%
_{R}^{\prime })\leq C\right\}  \notag \\
&\leq &\Pr \left\{ \mathbf{y}\in \mathcal{S}_{\text{C}}\right\}  \notag \\
&=&\int_{\mathcal{S}_{\text{C0}}}f(\mathbf{\tilde{n}}_{\text{v}})d\mathbf{%
\tilde{n}}_{\text{v}}  \notag \\
&\leq &\frac{C}{\pi ^{N_{\text{E}}}P_{\text{v}}^{N_{\text{E}}}}  \notag \\
&=&\left( \rho /\kappa (d)\right) ^{-N_{\text{B}}}\text{,}  \label{in2}
\end{eqnarray}%
where $f(\mathbf{\tilde{n}}_{\text{v}})$ is the probability density function
(pdf) of $\mathbf{\tilde{n}}_{\text{v}}$. The last inequality holds since%
\begin{eqnarray}
f(\mathbf{\tilde{n}}_{\text{v}}) &=&\frac{1}{\pi ^{N_{\text{E}}}P_{\text{v}%
}^{N_{\text{E}}}}\exp \left( -\frac{||\mathbf{\tilde{n}}_{\text{v}}||^{2}}{%
\tilde{\sigma}_{\text{v}}^{2}}\right)  \notag \\
&\leq &\frac{1}{\pi ^{N_{\text{E}}}P_{\text{v}}^{N_{\text{E}}}}\text{.}
\end{eqnarray}

\emph{2) }$\Pr \left\{ \text{vol}(\mathcal{S}_{R}^{\prime })>C\right\} $%
\emph{:} Since vol$(\mathcal{S}_{R}^{\prime })=d\cdot $vol$(\Lambda _{%
\mathbb{C}})$, we have%
\begin{equation}
\Pr \left\{ \text{vol}(\mathcal{S}_{R}^{\prime })>C\right\} =\Pr \left\{
\Delta (d)>\left( \rho /\kappa (d)\right) ^{-N_{\text{B}}}\right\} \text{.}
\label{in}
\end{equation}

From (\ref{Pd}), (\ref{in2}), (\ref{in}) and (\ref{Pr_inf}), as $N_{\text{B}%
}\rightarrow \infty $,%
\begin{equation}
F_{D}(d,P_{\text{v}})<O\left( \left( \frac{\rho }{\kappa (d)}\right) ^{-N_{%
\text{B}}}\right) \text{.}  \label{in1}
\end{equation}

From (\ref{Pd}), (\ref{in2}), (\ref{in}) and (\ref{Pr_Fin}), when $N_{\text{B%
}}$ is finite,%
\begin{equation}
F_{D}(d,P_{\text{v}})<\left( \frac{\rho }{\kappa (d)}\right) ^{-N_{\text{B}%
}}+\Upsilon \left( \frac{\rho }{\kappa (d)}\right) \text{.}
\end{equation}

\QEDA
\vspace{-5 mm}

\subsection{Proof of Theorem 2}

From (\ref{P_out}) and (\ref{uncer}), we have%
\begin{equation}
P_{\text{out}}(d)=F_{D}(d\text{, }P_{\text{v}})\text{.}
\end{equation}

Let $\rho =\varepsilon ^{-1/N_{\min }}\kappa (d)$, for arbitrarily small $%
\varepsilon >0$. We have%
\begin{equation}
\left( \rho /\kappa (d)\right) ^{-N_{\text{B}}}=\varepsilon ^{N_{\text{B}%
}/N_{\min }}\leq \varepsilon \text{.}  \label{E1}
\end{equation}

From Lemma \ref{lem1}, (\ref{E1}), and (\ref{NLP0}), if $P_{\text{v}}\geq
\rho ^{2}/\Phi ^{2N_{\text{B}}/N_{\text{E}}}$, we have%
\begin{equation}
F_{D}(d\text{, }P_{\text{v}})<\varepsilon +\Upsilon (\varepsilon
^{-1/N_{\min }})=O(\varepsilon )\text{,}
\end{equation}%
or equivalently,%
\begin{equation}
P_{\text{out}}(d)<O(\varepsilon )\text{.}
\end{equation}

\QEDA\vspace{-5mm}

\subsection{Proof of Lemma 4}

Recalling that%
\begin{equation}
R_{\max }(P_{\text{v}})=\max_{||\mathbf{v||}^{2}\leq P_{\text{v}}}\left\Vert
\mathbf{G}\mathbf{Z}\mathbf{v}\right\Vert \text{,}  \label{aabb}
\end{equation}%
\begin{equation}
r_{\text{eff}}(\Lambda _{\mathbb{C}})=\sqrt{N_{\text{B}}/(\pi e)}|\det ((%
\mathbf{GV}_{1}\mathbf{)}^{H}(\mathbf{GV}_{1}\mathbf{))|}^{\frac{1}{2N_{%
\text{B}}}}\text{.}
\end{equation}

From (\ref{R_max}), applying Cauchy--Schwarz inequality,%
\begin{equation}
R_{\max }^{2}(P_{\text{v}})=\lambda _{\max }P_{\text{v}}\leq P_{\text{v}%
}\left\Vert \mathbf{G}\mathbf{Z}\right\Vert _{F}^{2}\text{.}
\end{equation}

From Alice perspective, $\mathbf{G}\mathbf{Z}$ is a complex Gaussian random
matrix with i.i.d. components. Thus, $\left\Vert \mathbf{G}\mathbf{Z}%
\right\Vert _{F}^{2}$ can be expressed in terms of a Chi-squared random
variable:%
\begin{equation}
\left\Vert \mathbf{G}\mathbf{Z}\right\Vert _{F}^{2}=\frac{1}{2}\mathcal{X}%
^{2}\left( 2N_{\text{E}}\left( N_{\text{A}}-N_{\text{B}}\right) \right)
\text{.}  \label{RM1}
\end{equation}

According to (\ref{r1}), $r_{\text{eff}}(\Lambda _{\mathbb{C}})$ can be
expressed in terms of $N_{\text{B}}$ independent Chi-squared variables:%
\begin{equation}
r_{\text{eff}}(\Lambda _{\mathbb{C}})=\sqrt{N_{\text{B}}/(\pi e)}\left(
\prod_{j=1}^{N_{\text{B}}}\frac{1}{2}\mathcal{X}^{2}\left( 2(N_{\text{E}%
}-j+1)\right) \right) ^{\frac{1}{2N_{\text{B}}}}\text{.}  \label{RE2}
\end{equation}

Moreover, since $\mathbf{GV}_{1}$ and $\mathbf{G}\mathbf{Z}$ are mutually
independent \cite{Lukacs54}, $R_{\max }(P_{\text{v}})$ and $r_{\text{eff}%
}(\Lambda _{\mathbb{C}})$ are independent.

Then, we have%
\begin{align}
& \Pr \left\{ \dfrac{2R_{\max ,i}(P_{\text{v}})}{\sqrt{M}r_{\text{eff}%
,i}(\Lambda _{\mathbb{C}})}<x\right\}  \notag \\
& \geq \Pr \left\{ \frac{P_{\text{v}}\left\Vert \mathbf{G}\mathbf{Z}%
\right\Vert _{F}^{2}}{r_{\text{eff}}(\Lambda _{\mathbb{C}})^{2}}<\frac{x^{2}M%
}{4}\right\}  \notag \\
& =\Pr \left\{ \frac{\mathcal{X}^{2}\left( 2N_{\text{E}}\left( N_{\text{A}%
}-N_{\text{B}}\right) \right) }{\left( \prod_{j=1}^{N_{\text{B}}}\mathcal{X}%
^{2}\left( 2(N_{\text{E}}-j+1)\right) \right) ^{\frac{1}{N_{\text{B}}}}}<%
\frac{x^{2}MN_{\text{B}}}{4\pi eP_{\text{v}}}\right\}  \notag \\
& \overset{a}{\geq }\Pr \left\{ \dfrac{\mathcal{X}^{2}\left( 2N_{\text{E}%
}\left( N_{\text{A}}-N_{\text{B}}\right) \right) }{\dfrac{N_{\text{B}}}{%
\sum_{j=1}^{N_{\text{B}}}\dfrac{1}{\mathcal{X}^{2}\left( 2(N_{\text{E}%
}-j+1)\right) }}}<\dfrac{x^{2}MN_{\text{B}}}{4\pi eP_{\text{v}}}\right\}
\notag \\
& =\Pr \left\{ \sum_{j=1}^{N_{\text{B}}}\frac{\mathcal{X}^{2}\left( 2N_{%
\text{E}}\left( N_{\text{A}}-N_{\text{B}}\right) \right) }{\mathcal{X}%
^{2}\left( 2(N_{\text{E}}-j+1)\right) }<\frac{x^{2}MN_{\text{B}}^{2}}{4\pi
eP_{\text{v}}}\right\}  \notag \\
& \overset{b}{>}\prod\limits_{j=1}^{N_{\text{B}}}\Pr \left\{ \frac{\mathcal{X%
}^{2}\left( 2N_{\text{E}}\left( N_{\text{A}}-N_{\text{B}}\right) \right) }{%
\mathcal{X}^{2}\left( 2(N_{\text{E}}-j+1)\right) }\leq \frac{x^{2}MN_{\text{B%
}}}{4\pi eP_{\text{v}}}\right\}  \notag \\
& =\prod\limits_{j=1}^{N_{\text{B}}}\Pr \left\{ \mathcal{F}(2N_{\text{E}%
}\left( N_{\text{A}}-N_{\text{B}}\right) \text{, }2(N_{\text{E}}-j+1))\leq
g(x,j)\right\} \text{,}  \label{xxx}
\end{align}%
where $g(x,j)$ is given in (\ref{g_x}), and $\mathcal{F}(k_{1}$, $k_{2})$
represents an $\mathcal{F}$-distributed random variable with $k_{1}$ and $%
k_{2}$ degrees of freedom. ($a$) holds due to the inequality of geometric
and harmonic means. ($b$) holds by induction on the fact that if the
non-negative random variables $A_{i}$, $1\leq i\leq N$, are mutually
independent, given a constant $C>0$,%
\begin{align}
& \Pr \left\{ \sum_{i=1}^{N}A_{i}<C\right\} >\Pr \left\{ A_{1}\leq
C/N;\sum_{i=2}^{N}A_{i}\leq C(N-1)/N\right\}  \notag \\
& =\Pr \left\{ A_{1}\leq C/N\right\} \Pr \left\{ \sum_{i=2}^{N}A_{i}\leq
C(N-1)/N\right\} \text{.}
\end{align}

Since the cdf of $\mathcal{F}(k_{1}$, $k_{2})$ can be expressed using the
regularized incomplete beta function \cite{Imbeta}, the final expression of (%
\ref{xxx}) is given in (\ref{P_t}).

\QEDA\vspace{-5mm}

\subsection{Proof of Theorem 3}

From Alice perspective, $L_{i}$ is a function of $\mathbf{G}_{i}$. Since $%
\left\{ \mathbf{G}_{i}\right\} _{1}^{B}$ are mutually independent, $\left\{
L_{i}\right\} _{1}^{B}$ are mutually independent. From (\ref{PF_out}), we
have%
\begin{equation}
P_{\text{F,out}}(d,B)=\prod\limits_{i=1}^{B}\Pr \left\{ L_{i}<d\right\}
\text{.}  \label{122}
\end{equation}

We then evaluate $\Pr \left\{ L_{i}<d\right\} $. For simplicity, we remove
the index $i$. According to Theorem \ref{Th2}, with $P_{\text{v}%
}=\varepsilon ^{-2/N_{\min }}\kappa (d)^{2}/\Phi ^{2N_{\text{B}}/N_{\text{E}%
}}$, we have
\begin{equation}
\Pr (D<d)<O(\varepsilon )\text{.}
\end{equation}

We can upper bound $\Pr \left\{ L<d\right\} $ by%
\begin{align}
& \Pr \left\{ L<d\right\}  \notag \\
& =\Pr \{L<d|D\geq d\}\Pr \{D\geq d\}+\Pr \{L<d|D<d\}\Pr \{D<d\}  \notag \\
& \leq \Pr \left\{ L<D|D\geq d\right\} \Pr \{D\geq d\}+O(\varepsilon )
\notag \\
& \leq \Pr \left\{ L<D\right\} +O(\varepsilon )\text{.}  \label{d_D}
\end{align}

We then evaluate $\Pr \left\{ L<D\right\} $.%
\begin{align}
& \Pr \left\{ L<D\right\} =\Pr \{L<D|\Theta (P_{\text{v}})<\varepsilon \}\Pr
\{\Theta (P_{\text{v}})<\varepsilon \}  \notag \\
& +\Pr \{L<D|\Theta (P_{\text{v}})\geq \varepsilon \}\Pr \{\Theta (P_{\text{v%
}})\geq \varepsilon \}  \notag \\
& \leq \Pr \{L<D|\Theta (P_{\text{v}})<\varepsilon \}+\Pr \{\Theta (P_{\text{%
v}})\geq \varepsilon \}\text{,}  \label{L_E}
\end{align}%
where $\Theta (P_{\text{v}})$ is given in (\ref{RRR}).

We then evaluate the two terms in (\ref{L_E}), separately.

\emph{1)} $\Pr \left\{ L<D|\Theta (P_{\text{v}})<\varepsilon \right\} $\emph{%
:} Recalling that%
\begin{equation}
\mathbf{y}=\mathbf{GV}_{1}\mathbf{{{\mathbf{u+GZv}}}}\text{,}
\end{equation}%
\begin{equation}
\Lambda _{\text{F}}=\{\mathbf{G\mathbf{\mathbf{\mathbf{\mathbf{V}}}}}_{1}%
\mathbf{\mathbf{{{\mathbf{u}}}}},{\mathbf{{\mathbf{u}}}}\in \mathcal{\tilde{Q%
}}{^{N_{\text{B}}}\}}\text{.}
\end{equation}

Since $L=|S_{R_{\max }}\cap \Lambda _{\text{F}}|$, we begin by checking the
boundary of $\Lambda _{\text{F}}$. Let $\mathbf{O}$ be the center point of $%
\Lambda _{\text{F}}$. According to \cite{Ajtai02}, for the Gaussian random
lattice basis $\mathbf{G\mathbf{\mathbf{\mathbf{\mathbf{V}}}}}_{1}$, the
boundary of $\Lambda _{\text{F}}$ can be approximated by a sphere $S_{\text{%
F,S}}$ centered at $\mathbf{O}$ with radius $\sqrt{M}r_{\text{eff}}(\Lambda
_{\mathbb{C}})$, where $r_{\text{eff}}(\Lambda _{\mathbb{C}})$ is given in (%
\ref{r_eff}).

Given $\Theta (P_{\text{v}})<\varepsilon $ and $\varepsilon <1$, we have $%
\sqrt{M}r_{\text{eff}}(\Lambda _{\mathbb{C}})>2R_{\max }(P_{\text{v}})$. We
define a concentric sphere $S_{\text{F,C}}$ with radius $\sqrt{M}r_{\text{eff%
}}(\Lambda _{\mathbb{C}})-2R_{\max }(P_{\text{v}})$, where $R_{\max }(P_{%
\text{v}})$ is given in (\ref{R_max}). We then check when $L=D$ given $%
\Theta (P_{\text{v}})<\varepsilon $.

If $\mathbf{GV}_{1}\mathbf{{{\mathbf{u}}}}\in S_{\text{F,C}}$, using
triangle inequality, we have%
\begin{eqnarray}
||\mathbf{y-O||} &\leq &\left\Vert \mathbf{G\mathbf{\mathbf{V}}}_{1}\mathbf{{%
{\mathbf{u}}}}-\mathbf{O}\right\Vert +\left\Vert \mathbf{{{\mathbf{GZv}}}}%
\right\Vert  \notag \\
&\leq &\sqrt{M}r_{\text{eff}}(\Lambda _{\mathbb{C}})-R_{\max }(P_{\text{v}})%
\text{.}  \label{b1}
\end{eqnarray}%
We then check the locations of the $D$ elements in $S_{R_{\max }}\cap
\Lambda _{\mathbb{C}}$ (\ref{Di}), denoted by, $\mathbf{G\mathbf{\mathbf{V}}}%
_{1}\mathbf{u}_{t}^{\prime }$, $1\leq t\leq D$. Note that%
\begin{equation}
\left\Vert \mathbf{G\mathbf{\mathbf{V}}}_{1}\mathbf{u}_{t}^{\prime }-\mathbf{%
{{\mathbf{y}}}}\right\Vert \leq R_{\max }(P_{\text{v}})\text{.}  \label{b2}
\end{equation}%
From (\ref{b1}) and (\ref{b2}), using triangle inequality, for all $t$,%
\begin{equation}
\left\Vert \mathbf{G\mathbf{V}}_{1}\mathbf{u}_{t}^{\prime }-\mathbf{O}%
\right\Vert \leq \left\Vert \mathbf{y-O}\right\Vert +\left\Vert \mathbf{G%
\mathbf{\mathbf{V}}}_{1}\mathbf{u}_{t}^{\prime }-\mathbf{{{\mathbf{y}}}}%
\right\Vert \leq \sqrt{M}r_{\text{eff}}(\Lambda _{\mathbb{C}})\text{.}
\end{equation}%
Therefore, $S_{R_{\max }}\cap \Lambda _{\mathbb{C}}\subset \Lambda _{\text{F}%
}$, i.e., $L=D$.

If $\mathbf{G\mathbf{\mathbf{V}}}_{1}\mathbf{{{\mathbf{u}}}}\notin S_{\text{%
F,C}}$, there is a probability that $L<D$. Therefore, we have%
\begin{equation}
\Pr \left\{ L<D|\Theta (P_{\text{v}})<\varepsilon \right\} <\Pr \left\{
\mathbf{GV}_{1}\mathbf{{{\mathbf{u}}}}\notin S_{\text{F,C}}\right\} \text{.}
\label{EEE1}
\end{equation}%
Since $\mathbf{G\mathbf{\mathbf{V}}}_{1}\mathbf{{{\mathbf{u}}}}$ is
uniformly distributed over $S_{\text{F,S}}$, we have%
\begin{equation}
\Pr \left\{ \mathbf{GV}_{1}\mathbf{{{\mathbf{u}}}}\in S_{\text{F,C}}\right\}
=\frac{\text{vol}(S_{\text{F,C}})}{\text{vol}(S_{\text{F,S}})}=\left(
1-\Theta (P_{\text{v}})\right) ^{2N_{\text{B}}}>\left( 1-\varepsilon \right)
^{2N_{\text{B}}}  \label{EEE2}
\end{equation}

Based on (\ref{EEE1}) and (\ref{EEE2}), we have%
\begin{equation}
\Pr \left\{ L<D|\Theta (P_{\text{v}})<\varepsilon \right\} <1-\left(
1-\varepsilon \right) ^{2N_{\text{B}}}=O(\varepsilon )\text{.}  \label{pv1}
\end{equation}

\emph{2)} $\Pr \{\Theta (P_{\text{v}})\geq \varepsilon \}$\emph{:} Using
Lemma \ref{Th4} with $M\geq \varepsilon ^{-3-2/N_{\min }}\kappa (d)^{2}$, we
have%
\begin{align}
\Pr \left\{ \Theta (P_{\text{v}})<\varepsilon \right\} & \geq
\prod\limits_{j=1}^{N_{\text{B}}}B_{a,b(j)}\left( 1-\frac{b(j)}{%
ag(\varepsilon ,j)+b(j)}\right)  \notag \\
& \overset{a}{=}\prod\limits_{j=1}^{N_{\text{B}}}1-B_{b(j),a}\left( \frac{%
b(j)}{ag(\varepsilon ,j)+b(j)}\right)  \notag \\
& \overset{b}{=}\prod\limits_{j=1}^{N_{\text{B}}}\left( 1-O(\varepsilon ^{N_{%
\text{E}}-j+1})\right)  \notag \\
& >\left( 1-O(\varepsilon ^{N})\right) ^{N_{\text{B}}}\text{,}
\end{align}%
where $N=N_{\text{E}}-N_{\text{B}}+1$ and%
\begin{equation}
a=N_{\text{E}}(N_{\text{A}}-N_{\text{B}})\text{ and }b(j)=N_{\text{E}}-j+1%
\text{.}
\end{equation}%
$(a)$ and $(b)$ hold due to the facts that%
\begin{equation}
B_{a,b}(x)=1-B_{b,a}(1-x)\text{,}
\end{equation}%
\begin{equation}
B_{b(j),a}(x)=O(x^{b(j)})\text{, for }x\rightarrow 0\text{.}
\end{equation}

Consequently, we have%
\begin{equation}
\Pr \{\Theta (P_{\text{v}})\geq \varepsilon \}<1-\left( 1-O(\varepsilon
^{N})\right) ^{N_{\text{B}}}=O(\varepsilon ^{N})\text{.}  \label{pv2}
\end{equation}

By substituting (\ref{L_E}), (\ref{pv1}) and (\ref{pv2}) to (\ref{d_D}), we
have%
\begin{equation}
\Pr \left\{ L<d\right\} <O(\varepsilon )\text{.}  \label{tt1}
\end{equation}

From (\ref{122}) and (\ref{tt1}), if $M\geq \varepsilon ^{-3-2/N_{\min
}}\kappa (d)^{2}$ and $P_{\text{v}}=\varepsilon ^{-2/N_{\min }}\kappa
(d)^{2}/\Phi ^{2N_{\text{B}}/N_{\text{E}}}$, we have%
\begin{equation}
P_{\text{F,out}}(d,B)<O(\varepsilon ^{B})\text{.}
\end{equation}

\QEDA\vspace{-5mm}
\bibliographystyle{IEEEtran}
\bibliography{IEEEabrv,LIUBIB}

\end{document}